\documentclass [a4paper,11pt]{article}
\usepackage[english]{babel}
\usepackage[tbtags]{amsmath}
\usepackage{amssymb,amsthm,amsfonts}
\usepackage[T1]{fontenc}
\usepackage[latin1]{inputenc}
\usepackage{graphicx}
\usepackage{caption}
\usepackage{fancyhdr}
\usepackage{dsfont}
\usepackage{enumitem}
\usepackage[a4paper,top=2.5cm, bottom=2.5cm, left=2.5cm, right=2.5cm]{geometry}
\usepackage{soul}
\usepackage{titlesec}
\usepackage{bm}
\usepackage[symbol]{footmisc}

\allowdisplaybreaks[3]

\titlelabel{\thetitle.\quad}

\newcommand{\R}{\mathbb{R}}

\newcommand{\T}{\mathbb{T}}
\newcommand{\Sf}{\mathbb{S}^2}

\newcommand{\dd}{\mathrm{d}}            

\newcommand{\ac}{\'}

\newcommand{\init}{\textrm{in}}
\newcommand{\Ker}{\mathrm{Ker\hspace{0.5mm}}}
\newcommand{\Span}{\textrm{Span}}
\newcommand{\loc}{\textrm{loc}}

\newcommand{\boldF}{\mathbf{F}}
\newcommand{\boldf}{\mathbf{f}}
\newcommand{\boldL}{\mathbf{L}}
\newcommand{\boldM}{\pmb{\mu}}
\newcommand{\epsM}{M^{\varepsilon}}
\newcommand{\epsL}{L^{\varepsilon}}

\newcommand{\spaceLMR}{L^2\big(\R^3,\pmb{\mu}^{-1/2}\big)}
\newcommand{\spaceLM}{L_v^2\big(\pmb{\mu}^{-1/2}\big)}

\newcommand{\spaceLvM}{L_v^2\big(\langle v\rangle^{\gamma/2}\pmb{\mu}^{-1/2}\big)}

\newcommand{\specialsection}{%
  \titleformat{\section}[block]{\center\scshape\large}{\thesection.}{0.5em}{}
}

\newcommand{\addperiod}[1]{#1.}

\newcommand{\specialsubsection}{%
 \titleformat{\subsection}[runin]{\normalfont\large\bfseries}{\normalfont\thesubsection.}{0.5em}{\addperiod}
}

\newenvironment{abs}{\noindent\small\scshape{Abstract.}\upshape}

\graphicspath{{IMG/}}

\newtheorem{teor}{Theorem}[section]
\newtheorem{lemm}[teor]{Lemma}

\newtheorem{propos}[teor]{Proposition}
\newtheorem{propos*}{Proposition}
\newtheorem{corol}[teor]{Corollary}
\newtheorem{oss}[teor]{Remark} 

\makeatletter
\def\blfootnote{\xdef\@thefnmark{}\@footnotetext}
\makeatother

\begin{document}



\begingroup  
  \centering
  \textbf{\scshape{STABILITY OF THE SPECTRAL GAP\\FOR THE BOLTZMANN MULTI-SPECIES OPERATOR LINEARIZED\\ AROUND NON-EQUILIBRIUM MAXWELL DISTRIBUTIONS}}\\[1.5em]
   \footnotesize\scshape{ANDREA BONDESAN, LAURENT BOUDIN, MARC BRIANT, AND B\ac{E}R\ac{E}NICE GREC}\\[1.5em]

\endgroup

\begin{abs}
We consider the Boltzmann operator for mixtures with cutoff Maxwellian, hard potentials, or hard spheres collision kernels. In a perturbative regime around the global Maxwellian equilibrium, the linearized Boltzmann multi-species operator $\boldL$ is known to possess an explicit spectral gap $\lambda_{\boldL}$, in the global equilibrium weighted $L^2$ space. We study a new operator $\bold{\epsL}$ obtained by linearizing the Boltzmann operator for mixtures around local Maxwellian distributions, where all the species evolve with different small macroscopic velocities of order $\varepsilon$, $\varepsilon >0$. This is a non-equilibrium state for the mixture. We establish a quasi-stability property for the Dirichlet form of $\bold{\epsL}$ in the global equilibrium weighted $L^2$ space. More precisely, we consider the explicit upper bound that has been proved for the entropy production functional associated to $\boldL$ and we show that the same estimate holds for the entropy production functional associated to $\mathbf{\epsL}$, up to a correction of order $\varepsilon$.

\end{abs}

\specialsection
\specialsubsection


\section{Introduction}

\blfootnote{\textit{Date:} November 15, 2018 .}
\blfootnote{This work was partially funded by the French ANR-13-BS01-0004 project Kibord headed by L. Desvillettes. The
first and fourth authors have been partially funded by Universit\'e Sorbonne Paris Cit\'e, in the framework of the ``Investissements
d'Avenir'', convention ANR-11-IDEX-0005.}
The study of the convergence to the equilibrium for the Boltzmann equation has a long history (see the reviews \cite{Cer, CerIllPul, Vil, UkaYan}) and is at the core of the kinetic theory of gases. In the mono-species setting, the Boltzmann equation models the behaviour of a dilute gas composed of a large number of identical and monatomic particles which are supposed to only interact via microscopic binary collisions. It describes the time evolution of the gas distribution function $F=F(t,x,v)$ and reads, on $\R_+\times\T^3\times\R^3$,
\begin{equation}\label{MonospeciesBoltzmann}
\partial_t F +v\cdot\nabla_x F =Q(F,F).
\end{equation}

Without going into any details on the structure of the Boltzmann operator $Q$ for the moment, let us just mention that it is a quadratic operator which contains all the information about the collision process, prescribing how the microscopic interactions between particles take place. In particular, one of its main features is that it satisfies the celebrated $H$-theorem, which states that the entropy functional $H(F)=\int_{\R^3}F\log F\dd v$ decreases with respect to time and the gas reaches an equilibrium state whenever $Q(F,F)=0$, namely when $F$ has the form of a local Maxwellian distribution
\[
M(t,x,v)=\frac{c(t,x)}{\big(2\pi K_B T(t,x)\big)^{3/2}}\exp\left\{-\frac{|v-u(t,x)|^2}{2 K_B T(t,x)}\right\},\hspace{0.5cm} (t,x,v)\in\R_+\times\T^3\times\R^3,
\]
where $K_B$ is the Boltzmann constant.
Consequently, taking into account the influence of the transport operator in the left-hand side of (\ref{MonospeciesBoltzmann}), the $H$-theorem suggests that, for large time asymptotics, the solutions of the Boltzmann equation should converge towards a global equilibrium having the form of a uniform (in time and space) Maxwellian distribution
$M_{\infty}(v)=c_{\infty}/(2\pi K_B T_{\infty})^{3/2}\exp\left\{-\frac{|v-u_{\infty}|^2}{2 K_B T_{\infty}}\right\}$, with $c_{\infty}, T_{\infty}>0$ and $u_{\infty}\in\R^3$. Here, the quantities $c_{\infty}$, $u_{\infty}$ and $T_{\infty}$ respectively represent the number of particles, macroscopic velocity and temperature of the gas, and their values are prescribed by the initial datum. For the sake of simplicity, up to a translation and a dilation of the coordinate system, it is always possible to choose $M_{\infty}$ as being the global Maxwellian distribution 
\[
\mu(v)=\frac{c_{\infty}}{(2\pi)^{3/2}}e^{-\frac{|v|^2}{2}},\hspace{0.5cm} v\in\R^3.
\]

Starting from an initial datum close enough to $\mu$, one could expect that the solutions to the Boltzmann equation remain in a neighbourhood of the equilibrium and eventually converge to it in the large time asymptotic. A common strategy to tackle the problem of convergence towards the equilibrium is then to perform studies near the  global Maxwellian distribution, by investigating the behaviour of a small perturbation $f$ around $\mu$.

This theory was initiated by Hilbert \cite{Hil} and later on further developed by Carleman \cite{Car} and Grad \cite{Gra1}. In the perturbative regime around the global equilibrium $\mu$, the Boltzmann operator $Q$ can be rewritten as $Q(F,F)=\{Q(\mu,f)+Q(f,\mu)\}+Q(f,f)$, since $Q(\mu,\mu)=0$. Because $f$ is supposed to remain small, the dominant term in the previous relation is the so-called linearized Boltzmann operator
\[
L(f)=Q(\mu,f)+Q(f,\mu).
\]
The question which naturally arises concerns its spectral properties. In fact, as done in the nonlinear case for $Q$, it is possible to prove a linearized version of the $H$-theorem which asserts that the so-called entropy production functional (or Dirichlet form) $D(f)=\int_{\R^3}L(f)f\mu^{-1}\dd v$ is nonpositive and satisfies in particular the following fundamental coercivity estimate \cite{Car,Gra1,Gra3,Cer,CerIllPul,BarMou}:
\begin{equation}\label{SpectralMonospecies}
D(f)\leq -\lambda_L\Vert f-\pi_L(f)\Vert^2_{L_v^2(\mu^{-1/2})},
\end{equation}
for any $f\in L^2(\R^3,\mu^{-1/2})=\Big\{f : \R^3\to\R\ \mathrm{measurable}:\Vert f\Vert_{L_v^2(\mu^{-1/2})}=\int_{\R^3}f^2\mu^{-1}\dd v < +\infty\Big\}$. Here, $\pi_L$ is the orthogonal projection onto $\Ker L$ (the space of the equilibrium states of $L$), and $\lambda_L>0$ can be taken to be the smallest nonzero eigenvalue of $-L$, i.e. the spectral gap of the linearized Boltzmann operator. 

Consequently, in a purely linear setting where $\partial_t f +v\cdot\nabla_x f =L(f)$, the above inequality allows to formally deduce that the associated linearized entropy functional $\int_{\R^3}f^2\mu^{-1}\dd v$ is monotonically decreasing with respect to time, and the gas reaches an equilibrium state whenever $L(f)=0$. Proving that $L$ satisfies (\ref{SpectralMonospecies}) provides a powerful tool in estimating the asymptotic behaviour of solutions of the Boltzmann equation close to the equilibrium, since it allows to derive a quantitative information about the rate of relaxation of $F$ towards $\mu$.

The existence of a spectral gap $\lambda_L$ for $L$ has been known for long and many results have been proved in this regard over the past sixty years, starting with Carleman and Grad themselves \cite{Car, Gra1, Gra3} in the case of Maxwellian, hard potentials and hard spheres collision kernels with cutoff. The idea behind these methods is to apply Weyl's theorem to $L$, written as a compact perturbation of a multiplicative operator, and they are thus non-constructive. Consequently, a major inconvenience of this strategy is that it does not provide any information about the width of the spectral gap, or its dependence on the initial datum and on the physical quantities appearing in the problem (in particular its sensitivity to perturbations of the collision kernel). 

The derivation of explicit estimates for the spectral gap of the linearized Boltzmann operator is fundamental to construct a quantitative theory near the equilibrium and to obtain explicit convergence rates to the equilibrium. In fact, when proving that solutions to the Boltzmann equation in the perturbative regime have an exponential trend of relaxation of order $\mathcal{O}\big(e^{-t/\tau}\big)$ towards the equilibrium \cite{Mou2,MouNeu}, one can use the information provided by (\ref{SpectralMonospecies}) in order to link the rate $1/\tau$ to the spectral gap $\lambda_L$ and recover an explicit convergence speed. This feature has been of great importance in the kinetic theory of gases, starting from the studies of Boltzmann himself. Indeed, the validity of the Boltzmann equation breaks for very large time (see the discussion in \cite[Chapter 1, Section 2.3]{Vil}). It is therefore crucial to obtain constructive quantitative informations on the time scale of the convergence, in order to show that this time scale is much smaller than the time scale of validity of the model.

In the mono-species setting, the first outcome about explicit estimates for the spectral gap $\lambda_L$ has been obtained by Bobylev in \cite{Bob1, Bob2} for the case of Maxwell molecules, using Fourier methods to obtain a complete and explicit diagonalization of $L$. More recently, a systematic derivation of explicit coercivity estimates for the linearized collision operator has been established \cite{BarMou, Mou1, MouStr} for a wide class of collision kernels (including cutoff hard potentials and hard spheres).

Starting from these results, in the past few years, the question of convergence towards the equilibrium has also been addressed in the case of multi-species kinetic models. Many Boltzmann-like equations describing the time evolution of a mixture of different gases can be found in the literature. To mention a few of them, we can refer to \cite{Sir,Gre,Mor} for the first kinetic models for mixtures, to \cite{GarSanBre,AndAokPer,BruPavSch} for some BGK-type models, and to \cite{RosMaz,Sir,BouDesLetPer,RosSpi,DesMonSal} for multi-species kinetic models where also chemical reactions are taken into account. Moreover, we recall the recent works \cite{BouGrePavSal,DauJunMouZam,BriDau,Bri}, where some of these models are mathematically investigated.

\vspace{1mm}
In this work, we focus on the multi-species setting. The equation we consider \cite{DesMonSal,BouGreSal,BouGrePav} can be seen as the counterpart of (\ref{MonospeciesBoltzmann}) in the case of gaseous mixtures, and shall be presented in detail in the next section. We only recall here that it models a mixture of $N$ different ideal and monatomic gases, prescribing the time evolution of the multi-component gas distribution function $\boldF=(F_1,\ldots, F_N)$. It reads, on $\R_+\times\T^3\times\R^3$,
\begin{equation}\label{MultispeciesBoltzmann}
\partial_t \boldF +v\cdot\nabla_x\boldF=\mathbf{Q}(\boldF,\boldF).
\end{equation}
As in the mono-species framework, it is possible to prove \cite{DesMonSal} a multi-species version of the $H$-theorem for the Boltzmann operator $\mathbf{Q}$, which suggests in this case that the global equilibrium of the mixture can be taken, after rescaling of the coordinate system, to be the global Maxwellian distribution $\pmb{\mu}=(\mu_1,\ldots,\mu_N)$, with
\[
\mu_i(v)= c_{i,\infty}\left(\frac{m_i}{2\pi}\right)^{3/2}e^{-m_i\frac{|v|^2}{2}},\hspace{0.5cm} v\in\R^3,
\] 
for any $1\leq i\leq N$. Here, we have denoted by $(m_i)_{1\leq i\leq N}$ and $(c_{i,\infty})_{1\leq i\leq N}\in\R^N$ the atomic masses and numbers of particles of each species in the gaseous mixture. 

In order to investigate the convergence of solutions of (\ref{MultispeciesBoltzmann}) towards the equilibrium $\pmb{\mu}$, we can study the behaviour of a small perturbation $\boldf=(f_1,\ldots,f_N)$ around $\pmb{\mu}$. This approach raises the same kind of questions explained in our earlier discussion on the mono-species setting. In particular, the analysis of the large time asymptotic relaxation for solutions of the Boltzmann multi-species equation shall eventually lead to deal again with the problem of explicitly quantifying the rate of convergence to the equilibrium. 

The study of explicit coercivity estimates for the linearized Boltzmann multi-species operator $\boldL(\boldf)=\mathbf{Q}(\pmb{\mu},\boldf)+\mathbf{Q}(\boldf,\pmb{\mu})$ has been very recently tackled, and the existence of a spectral gap for this operator has been shown in \cite{DauJunMouZam} for mixtures where the atomic masses of the different species are the same ($m_i =m_j$ for any $1\leq i,j\leq N$), and in \cite{BriDau}, in the general case of species with different atomic masses. Both results are valid for Maxwellian, hard potentials and hard spheres collision kernels with cutoff.

In the multi-species framework, the complexity of the cross interactions between particles of the different gaseous components leads to new interesting physical phenomena \cite{BouGreSal,BouGrePav}. These can be described considering perturbative regimes which are alternative to the common one discussed above, around the global equilibrium $\pmb{\mu}$. Inspired by \cite{BouGreSal}, we here investigate a new perturbative setting around local Maxwellian distributions $\mathbf{\epsM}=(\epsM_1,\ldots,\epsM_N)$ with small macroscopic velocities of order $\varepsilon$, different for each species of the mixture. More precisely, for any $1\leq i\leq N$, we can write $\epsM_i$ as
\begin{equation}\label{LocMaxDis}
\epsM_i(t,x,v)=c_{i,\infty}\left(\frac{m_i}{2\pi}\right)^{3/2}\exp\left\{-m_i\frac{|v-\varepsilon u_i(t,x)|^2}{2}\right\},\hspace{0.5cm} (t,x,v)\in\R_+\times\T^3\times\R^3.
\end{equation}
We note in particular that the parameter $\varepsilon$ only translates the smallness of the macroscopic velocities of the different species, and not the order of the perturbation around $\mathbf{\epsM}$.

We shall see that these Maxwellians are not a local equilibrium for the mixture, whereas they were in the mono-species setting. This feature induces additional difficulties if one wants to analyze the spectrum of the linearized operator $\boldL^{\varepsilon}(\boldf)=\mathbf{Q}(\mathbf{\epsM},\boldf)+\mathbf{Q}(\boldf,\mathbf{\epsM})$, because $\boldL^{\varepsilon}$ does not satisfy any self-adjointness property and it appears problematic to determine the space of equilibrium states $\Ker \mathbf{\epsL}$. Consequently, we did not manage to give a detailed description of the spectrum of $\mathbf{\epsL}$ and show explicit spectral gap estimates, using the same methods developed in \cite{DauJunMouZam,BriDau}. 

Instead of trying to recover a full spectral gap property, our idea is to establish an explicit estimate for an upper bound of the entropy production functional associated to $\boldL^{\varepsilon}$. More precisely, we prove that the estimate known to exist for the Dirichlet form of $\boldL$ \cite{DauJunMouZam,BriDau} remains valid for the Dirichlet form of $\mathbf{\epsL}$, up to a correction term of order $\varepsilon$. The strategy exploits the fact that the non-equilibrium state $\mathbf{\epsM}$ is close to the global equilibrium $\pmb{\mu}$ up to a factor of order $\varepsilon$, and the same property holds for $\mathbf{\epsL}$ and $\boldL$. Specifically, we analyze the structure of the operator $\mathbf{\epsL}-\boldL$, showing that it can be seen as the sum of a multiplicative operator and an integral one. To deal with the latter, we exhibit its kernel structure by writing down a Carleman representation and deduce useful pointwise estimates. Due to the mixing effect in the exponential decay rates among the cross-interactions between species (each species has different atomic mass $m_i$ and different macroscopic velocity $\varepsilon u_i$), we need to carefully treat the Maxwellian weights by deriving accurate $L^1$ controls in the velocity variable. In the end, the operator $\mathbf{\epsL}-\boldL$ is bounded from above by a positive quantity of order $\varepsilon$, producing the correction term in the estimate of the Dirichlet form of $\mathbf{\epsL}$.

\vspace{2mm}
 The work is organized as follows. We first recall the multi-species model and give a detailed description of its main features, introducing the linearizations of the Boltzmann multi-species operator which we intend to study. Next, we recall the fundamental properties of the linearized Boltzmann multi-species operator and state our main result. The last section is dedicated to the proof of this result.


\section{A kinetic model for mixtures}

The evolution of a dilute gas composed of $N$ different species of chemically non-reacting mono-atomic particles having atomic masses $(m_i)_{1\leq i\leq N}$ can be modelled by the following equations of Boltzmann type, for any $1\leq i\leq N$,

\begin{equation}\label{BoltzmannEquation}
\partial_t F_i(t,x,v)+v\cdot\nabla_x F_i(t,x,v)=Q_i(\boldF,\boldF)(t,x,v),\quad t\geq 0,\ x\in\T^3,\ v\in\R^3,
\end{equation}
The vector $\boldF=(F_1,\ldots,F_N)$ denotes the unknown vector-valued distribution function describing the mixture and each $F_i$ is the distribution function corresponding to the $i$-th species, which depends on time $t$, position $x$ and velocity $v$. Henceforth, vectors and vector-valued functions will be denoted by bold letters, while the corresponding indexed letters will indicate their components. For example, $\mathbf{W}$ represents the vector or vector-valued function $(W_1,\ldots ,W_N)$.

\vspace{1mm}
 The Boltzmann multi-species operator $\mathbf{Q}(\boldF,\boldF)=\big(Q_1(\boldF,\boldF),\ldots,Q_N(\boldF,\boldF)\big)$ describes the way the mixture molecules interact. The operator only acts on the velocity variable and it is defined component-wise, for any $i$, by
\[
Q_i(\boldF,\boldF)(v)=\sum_{j=1}^N Q_{ij}(F_i,F_j)(v)=\sum_{j=1}^N\int_{\R^3\times\Sf}B_{ij}(|v-v_*|,\theta)(F_i^{\prime}F_j^{\prime *}-F_i F_j^*)\dd\sigma\dd v_*,\hspace{0.5cm} v\in\R^3.
\]
In the previous equality, we used the standard shorthand notations $F_i^{\prime}=F_i(v^{\prime})$, $F_i=F_i(v)$, $F_j^{\prime *}=F_j(v_*^{\prime})$, $F_j^*=F_j(v_*)$. The pre-collisional velocities $v^{\prime}$ and $v_*^{\prime}$ are given in terms of the post-collisional velocities $v$ and $v_*$ by the elastic collision rules

\[
\left\{\begin{array}{l}
v^{\prime}=\frac{1}{m_i+m_j}(m_i v+m_j v_*+m_j|v-v_*|\sigma),\\
\\
v_*^{\prime}=\frac{1}{m_i+m_j}(m_i v+m_j v_*-m_i|v-v_*|\sigma),
\end{array}\right.
\]
where $\sigma\in\Sf$ is a parameter whose existence is ensured by the conservation of microscopic momentum and kinetic energy
\begin{equation}\label{ConservationProperties}
m_i v^{\prime}+m_j v_*^{\prime}=m_i v+m_j v_*,\qquad
\frac{1}{2}m_i|v^{\prime}|^2+\frac{1}{2}m_j|v_*^{\prime}|^2=\frac{1}{2}m_i|v|^2+\frac{1}{2}m_j|v_*|^2.
\end{equation}

\vspace{1mm}
 The collisional cross-sections $B_{ij}$ are nonnegative functions of $|v-v_*|$ and the cosine of the deviation angle $\theta\in[0,\pi]$ between $v-v_*$ and $\sigma\in\Sf$. They encode all the information about the microscopic interactions between molecules in the mixture, and their choice is essential when studying the properties of the Boltzmann operator. In this work, we focus on cutoff Maxwellian, hard potential and hard spheres collision kernels. Namely, let us make the following assumptions on each $B_{ij}$, $i$ and $j$ being fixed.

\begin{enumerate}
\item[(H1)] It satisfies a symmetry property with respect to the interchange of both species $i$ and $j$
\[
B_{ij}(|v-v_*|,\cos\theta)=B_{ji}(|v-v_*|,\cos\theta),\hspace{0.5cm} \forall\ v,v_*\in\R^3, \forall\ \theta\in\R.
\]
\item[(H2)] It decomposes into the product of a kinetic part $\Phi_{ij}\geq 0$ and an angular part $b_{ij}\geq 0$, namely
\[
B_{ij}(|v-v_*|,\cos\theta)=\Phi_{ij}(|v-v_*|)b_{ij}(\cos\theta),\hspace{0.5cm} \forall\ v,v_*\in\R^3, \forall\ \theta\in\R.
\]
\item[(H3)] The kinetic part has the form of hard or Maxwellian ($\gamma = 0$) potential, i.e.
\[
\Phi_{ij}(|v-v_*|)=C_{ij}^{\Phi}|v-v_*|^{\gamma},\quad C_{ij}^{\Phi}>0,\quad \gamma\in [0,1],\hspace{0.5cm} \forall\ v,v_*\in\R^3.
\]
\item[(H4)] For the angular part, we consider a strong form of Grad's angular cutoff \cite{Gra1}. We assume that there exists a constant $C >0$ such that
\[
0<b_{ij}(\cos\theta)\leq C|\sin\theta||\cos\theta|,\qquad b_{ij}^{\prime}(\cos\theta)\leq C,\qquad \theta\in [0,\pi].
\]
Furthermore, we assume that
\[
\inf_{\sigma_1,\sigma_2\in\Sf}\int_{\Sf}\min\{b_{ii}(\sigma_1\cdot\sigma_3),b_{ii}(\sigma_2\cdot\sigma_3)\}\dd \sigma_3 >0.
\]
\end{enumerate}

\vspace{1mm}
 From now on, let us use the following notations. Consider a positive measurable vector-valued function $\mathbf{W}=(W_1,\ldots,W_N):\R^3\to(\R_+^*)^N$ in the variable $v$. For any $1\leq i\leq N$, we define the weighted Hilbert space $L^2(\R^3, W_i)$ by introducing the scalar product and norm
\[
\langle f_i,g_i\rangle_{L_v^2(W_i)}=\int_{\R^3}f_i g_i W_i^2\dd v,\quad \| f_i\|^2_{L_v^2(W_i)}=\langle f_i,f_i\rangle_{L_v^2(W_i)},\quad \forall\ f_i,g_i\in L^2(\R^3,W_i).
\]
With this definition, we say that $\boldf=(f_1,\ldots, f_N):\R^3\to\R^N\in L^2(\R^3, \mathbf{W})$ if $f_i :\R^3\to\R\in L^2(\R^3,W_i)$, and we associate to $L^2(\R^3,\mathbf{W})$ the natural scalar product and norm
\[
\langle\boldf,\mathbf{g}\rangle_{L_v^2(\mathbf{W})}=\sum_{i=1}^N\langle f_i,g_i\rangle_{L_v^2(W_i)},\quad \|\boldf\|_{L_v^2(\mathbf{W})}=\left(\sum_{i=1}^N\| f_i\|^2_{L_v^2(W_i)}\right)^{1/2}.
\]
Moreover, for the sake of simplicity, we define the shorthand notation
\[
\langle v\rangle =\Big(1+|v|^2\Big)^{1/2},
\]
not to be confused with the usual Euclidean scalar product $\langle \cdot ,\cdot \rangle$ in $\R^3$, also used in the upcoming computations. 

\vspace{2mm}
Let us now recall the main properties of $\mathbf{Q}$, which can be found in \cite{CerIllPul, DesMonSal, BouGreSal, DauJunMouZam, BriDau}.\hspace{1mm}Using the symmetries of the collision operator combined with classical changes of variables $(v,v_*)\mapsto(v^{\prime},v_*^{\prime})$ and $(v,v_*)\mapsto(v_*,v)$, we can obtain the weak formulations of $Q_{ij}(F_i,F_j)$

\[
\hspace{-2.5cm}\int_{\R^3}Q_{ij}(F_i,F_j)(v)\psi_i(v)\dd v = \int_{\R^6\times\Sf}B_{ij}(|v-v_*|,\cos\theta)F_i F_j^*\big(\psi_i^{\prime}-\psi_i\big)\dd\sigma\dd v_*\dd v,
\]
\vspace{1mm}

\[
\hspace{-6cm}\int_{\R^3}Q_{ij}(F_i,F_j)(v)\psi_i(v)\dd v + \int_{\R^3}Q_{ji}(F_j,F_i)(v)\psi_j(v)\dd v =
\]
\[
\hspace{3cm}-\frac{1}{2}\int_{\R^6\times\Sf}B_{ij}(|v-v_*|,\cos\theta)\big(F_i^{\prime}F_j^{\prime *}-F_i F_j^*\big)\big(\psi_i^{\prime}+\psi_j^{\prime *}-\psi_i-\psi_j^*\big)\dd\sigma\dd v_*\dd v,
\]
for any vector-valued function $\pmb{\psi}=(\psi_1,\ldots,\psi_N):\R^3\to \R^N$ such that the integrals on the left-hand side of the above equalities are well-defined. These relations allow to deduce the conservation properties of the Boltzmann multi-species operator. More precisely, the equality

\[
\sum_{i=1}^N\int_{\R^3} Q_i(\boldF,\boldF)(v) \psi_i(v) \dd v=0
\]
holds if and only if $\pmb{\psi}$ is a collision invariant of the mixture, namely 
\[
\pmb{\psi}\in\textrm{Span}\big\{\mathbf{e}_1,\ldots,\mathbf{e}_N,v_1\mathbf{m},v_2\mathbf{m},v_3\mathbf{m},|v|^2\mathbf{m}\big\},
\] 
where $\mathbf{e}_k=(\delta_{ik})_{1\leq i\leq N}$ for $1\leq k\leq N$, and $\mathbf{m}=(m_1,\ldots,m_N)$. In particular, from (\ref{BoltzmannEquation}) and the conservation properties of $\mathbf{Q}$, we immediately get that the quantities

\begin{equation}\label{ConservationQuantities}
\begin{aligned}
c_{i,\infty}=\int_{\R^3\times\T^3}F_i(t,x,v)\dd x\dd v,\hspace{0.5cm} 1\leq i\leq N,\hspace{3cm}\\
u_{\infty}=\frac{1}{\rho_{\infty}}\sum_{i=1}^N\int_{\R^3\times\T^3}m_i v F_i\dd x\dd v,\hspace{0.5cm} \theta_{\infty}=\frac{1}{3\rho_{\infty}}\sum_{i=1}^N\int_{\R^3\times\T^3}m_i|v-u_{\infty}|^2F_i\dd x\dd v,
\end{aligned}
\end{equation}
are preserved for any time $t\geq 0$. Here, $c_{i,\infty}$ stands for the total number of particles of species $i$ and we have also defined the total mass of the mixture $\rho_{\infty}=\sum_{i=1}^N m_i c_{i,\infty}$, its total momentum $\rho_{\infty}u_{\infty}$ and its total energy $\frac{3}{2}\rho_{\infty}\theta_{\infty}$.

\vspace{1mm}
 Moreover, the operator $\mathbf{Q}$ satisfies a multi-species version of the $H$-theorem, see \cite[Proposition 1]{DesMonSal}. It states, in particular, that the only solutions of the equation $\mathbf{Q}(\boldF,\boldF)=0$ are distribution functions $\boldF=(F_1,\ldots,F_N)$ such that their components take the form of local Maxwellians with a common macroscopic velocity, i.e. for any $1\leq i\leq N$ and $(t,x,v)\in\R_+\times\T^3\times\R^3$,

\begin{equation}\label{LocalEquilibrium}
F_i(t,x,v) = c_{i,\loc}(t,x)\left(\frac{m_i}{2\pi K_B \theta_{\loc}(t,x)}\right)^{3/2}\exp\left\{-m_i\frac{|v-u_{\loc}(t,x)|^2}{2 K_B \theta_{\loc}(t,x)}\right\}.
\end{equation}
Here, we have defined the following local quantities:

\[
\begin{aligned}
c_{i,\loc}(t,x)=\int_{\R^3}F_i(t,x,v)\dd v,\hspace{0.3cm}1\leq i\leq N,\hspace{0.8cm}\rho_{\loc}(t,x)=\sum_{i=1}^N m_i c_{i,\loc}(t,x)\hspace{2cm}\\
u_{\loc}(t,x)=\frac{1}{\rho_{\loc(t,x)}}\sum_{i=1}^N\int_{\R^3}m_i v F_i\dd v,\hspace{0.8cm} \theta_{\loc}(t,x)=\frac{1}{3\rho_{\loc}(t,x)}\sum_{i=1}^N\int_{R^3}m_i|v-u_{\loc}|^2F_i\dd v.
\end{aligned}
\]

It is worth noting that the local Maxwellian distributions (\ref{LocalEquilibrium}) constitute a local equilibrium for the mixture, since $\mathbf{Q}(\boldF,\boldF)=0$. This is not the case for the local Maxwellians $\epsM_i$ introduced in (\ref{LocMaxDis}), which in contrary form a non-equilibrium state for the mixture because of the different macroscopic velocities $\varepsilon u_i$.

\vspace{1mm}
 In particular, since we work in $\T^3$, we can use the fact that the quantities (\ref{ConservationQuantities}) are conserved with respect to time, to deduce that the only global equilibrium of the mixture, i.e. the unique stationary solution $\boldF=(F_1,\ldots,F_N)$ to (\ref{BoltzmannEquation}), is given for any $1\leq i\leq N$ by the global Maxwellians

\[
F_i(v) = c_{i,\infty}\left(\frac{m_i}{2\pi K_B \theta_{\infty}}\right)^{3/2}\exp\left\{-m_i\frac{|v-u_{\infty}|^2}{2 K_B \theta_{\infty}}\right\},\hspace{0.5cm} v\in\R^3.
\]
Up to a translation and a dilation of the coordinate system, we can assume that $u_{\infty}=0$ and $K_B\theta_{\infty}=1$, so that the only global equilibrium $\boldM=(\mu_1,\ldots,\mu_N)$ becomes

\begin{equation}\label{GlobalEquilibrium}
\mu_i(v)= c_{i,\infty}\left(\frac{m_i}{2\pi}\right)^{3/2}e^{-m_i\frac{|v|^2}{2}},\hspace{0.5cm} v\in\R^3,
\end{equation}
for any $i$. In what follows, we shall always work in the weighted Hilbert space $\spaceLMR$, where we use the notation $\pmb{\mu}^{-1/2}=\big(\mu_1^{-1/2},\ldots,\mu_N^{-1/2}\big)$.

\vspace{1mm}
 As already mentioned in the introduction, in order to study the Cauchy problem (\ref{BoltzmannEquation}) and the convergence of the solution to equilibrium, it is common \cite{DesMonSal, BouGrePavSal,DauJunMouZam, BriDau} to consider a perturbative regime where each distribution function $F_i$ is close to the global equilibrium given by (\ref{GlobalEquilibrium}). More precisely, when writing $F_i(t,x,v)=\mu_i(v)+f_i(t,x,v)$ for any $1\leq i\leq N$, where $f_i$ is a small perturbation of the global equilibrium $\mu_i$, equations (\ref{BoltzmannEquation}) can be rewritten in the form of the following perturbed multi-species Boltzmann system, set in $\R_+\times\T^3\times\R^3$,
\[
\partial_t \boldf+v\cdot\nabla_x \boldf=\boldL(\boldf)+\mathbf{Q}(\boldf,\boldf).
\]
The linearized Boltzmann multi-species operator is given by $\boldL=(L_1,\ldots,L_N)$, with
\begin{equation}\label{GlobalLinearized1}
L_i(\boldf)=\sum_{j=1}^N L_{ij}(f_i,f_j)=\sum_{j=1}^N \Big(Q_{ij}(\mu_i,f_j)+Q_{ij}(f_i,\mu_j)\Big)
\end{equation} 
for any $i$, and its main properties will be recalled in the next section.\hspace{1mm}This operator has been intensively studied in the past few years, and the main results about his kernel and his spectral properties have been obtained in \cite{DauJunMouZam} in the case of species having the same masses $m_i=m_j$ and in \cite{BriDau} in the general case of species with different masses.

\vspace{2mm}
In this work, we consider a new perturbative regime around the local Maxwellian distributions $\mathbf{\epsM}=(\epsM_1,\ldots,\epsM_N)$, given for any $1\leq i\leq N$ by
\begin{equation}\label{EpsilonMaxwellian}
\epsM_i(t,x,v)= c_{i,\infty}\left(\frac{m_i}{2\pi}\right)^{3/2}\exp\left\{-m_i\frac{|v-\varepsilon u_i(t,x)|^2}{2}\right\},\hspace{0.5cm} (t,x,v)\in\R_+\times\T^3\times\R^3,
\end{equation}
with small macroscopic velocities depending on the parameter $\varepsilon\ll 1$.\hspace{1mm}We again stress the fact that the distribution functions $\epsM_i$ are local in time and space like the Maxwellians given in (\ref{LocalEquilibrium}), but they are not a local equilibrium for the Boltzmann operator $\mathbf{Q}=(Q_1,\ldots,Q_N)$. In fact, the macroscopic velocities are different for each species of the mixture (since $u_i\neq u_j$ for $i\neq j$), implying that $\mathbf{Q}(\mathbf{\epsM},\mathbf{\epsM})$ is not zero anymore. Nevertheless, it is also important to note that the macroscopic velocities we consider are small (of order $\varepsilon$), and this requirement is crucial to recover our result, since it shall be used to show that $\mathbf{\epsM}$ is close to $\pmb{\mu}$ up to an order $\varepsilon$. Thus, even if $\mathbf{\epsM}$ is a non-equilibrium state, it is still not far away from the global equilbrium of the mixture.

\vspace{1mm}
 If we now assume that each distribution function $F_i$ is close to a local Maxwellian of type (\ref{EpsilonMaxwellian}), namely $F_i(t,x,v)=\epsM_i(t,x,v)+f_i(t,x,v)$ for all $1\leq i\leq N$, with $\boldf=(f_1,\ldots,f_N)$ being again a small perturbation, equations (\ref{BoltzmannEquation}) posed on $\R_+\times\mathbb{T}^3\times\R^3$ can be rewritten as
\begin{equation}\label{EpsilonSystem}
\partial_t\mathbf{\epsM}+\partial_t \boldf+v\cdot\nabla_x \mathbf{\epsM}+v\cdot\nabla_x \boldf=\mathbf{Q}(\mathbf{\epsM},\mathbf{\epsM})+\mathbf{\epsL}(\boldf)+\mathbf{Q}(\boldf,\boldf).
\end{equation}
The operator $\mathbf{\epsL}=(\epsL_1,\ldots,\epsL_N)$ is the Boltzmann multi-species operator linearized around the local Maxwellian distributions (\ref{EpsilonMaxwellian}), and it is given by

\begin{equation}\label{LocalLinearized1}
\epsL_i(\boldf)=\sum_{j=1}^N \epsL_{ij}(f_i,f_j)=\sum_{j=1}^N \Big(Q_{ij}(\epsM_i,f_j)+Q_{ij}(f_i,\epsM_j)\Big),
\end{equation} 
for any $i$.

Before going into details of our main result, we here note that, for velocities $v-\varepsilon u_i$ and $v_* -\varepsilon u_j$, the microscopic conservation properties (\ref{ConservationProperties}) do not hold anymore and consequently ${\epsM_i}^{\prime}{\epsM_j}^{\prime *}\neq\epsM_i {\epsM_j}^*$. Therefore, the operator $\mathbf{\epsL}$ exhibits no good self-adjointness properties and this peculiarity will force us to develop an alternative strategy to the one presented in \cite{DauJunMouZam} and \cite{BriDau} in order to investigate its spectral properties. 

We also mention that looking at a perturbative regime around non-equilibrium Maxwellians $\epsM_i$ to recover system (\ref{EpsilonSystem}) can be of particular interest in the investigation of the asymptotic approximation of solutions to the Maxwell-Stefan diffusion equations, in the spirit of \cite{BouGreSal, BouGrePav}.


\section{Main result} 

This section is dedicated to the presentation of our main theorem. Due to the choice of the linearization of $\mathbf{Q}$ around $\mathbf{\epsM}$, the standard tools developed to perform the spectral analysis of $\boldL$ are no more valid for the operator $\mathbf{\epsL}$. This issue will be overcome by connecting $\mathbf{\epsL}$ to $\boldL$, in order to take advantage of the results known about the spectrum of $\boldL$. We are able to show that the Dirichlet form of $\mathbf{\epsL}$ in the space $\spaceLMR$ is upper-bounded by a negative term coming from the spectral gap of $\boldL$ in $\spaceLMR$, plus a positive contribution of order $\varepsilon$, which takes into account the fact that we are working in a space where $\mathbf{\epsL}$ is not self-adjoint.

 We begin by recalling the fundamental result derived in \cite{BriDau} about the existence of a spectral gap for $\boldL$ in the space $\spaceLMR$.
To do this, we first present the structure of $\Ker\boldL$. 

\vspace{1mm}
The linearized Boltzmann multi-species operator $\boldL$ defined by (\ref{GlobalLinearized1}) is a closed self-adjoint operator in $\spaceLMR$. Its kernel is described by an orthonormal basis $\big(\pmb{\phi}^{(i)}\big)_{1\leq i\leq N+4}$ in $\spaceLMR$, that is
\[
\Ker\boldL=\Span\big\{\pmb{\phi}^{(1)},\ldots,\pmb{\phi}^{(N+4)}\big\},
\]
where
\[
\left\{\begin{array}{ll}
\displaystyle \pmb{\phi}^{(k)}=\mu_k\mathbf{e}_k, \quad 1\leq k\leq N,\\
\\
\displaystyle \pmb{\phi}^{(N+l)}=\frac{v_l}{\left(\sum_{j=1}^N m_j c_{j,\infty}\right)^{1/2}}\big(m_i \mu_i\big)_{1\leq i\leq N}, \quad 1\leq l\leq 3,\\
\displaystyle \pmb{\phi}^{(N+4)}=\frac{1}{\left(\sum_{j=1}^N c_{j,\infty}\right)^{1/2}}\left(\frac{m_i|v|^2-3}{\sqrt{6}} \mu_i\right)_{1\leq i\leq N}.
\end{array}\right.
\]
Using this basis, for any $\boldf\in\spaceLMR$ we can define its orthogonal projection onto $\Ker\boldL$ in $\spaceLMR$ as
\[
\pi_{\boldL}(\boldf)(v)=\sum_{k=1}^{N+4}\big\langle\boldf ,\pmb{\phi}^{(k)}\big\rangle_{\spaceLM}\pmb{\phi}^{(k)}(v).
\]

Then, the linearized Boltzmann multi-species operator satisfies the following fundamental theorem, proved in \cite{BriDau}.

\begin{teor}\label{SpectralGapL}
Let the collision kernels $B_{ij}$ satisfy assumptions $\mathrm{(H1)-(H4)}$. Then there exists an explicit constant $\lambda_{\boldL}>0$ such that, for all $\boldf\in\spaceLMR$,

\[
\big\langle\boldL(\boldf),\boldf\big\rangle_{\spaceLM}\leq -\lambda_{\boldL}\|\boldf-\pi_{\boldL}(\boldf)\|_{\spaceLvM}^2.
\]
The constant $\lambda_{\boldL}$ only depends on $N$, the different masses $(m_i)_{1\leq i\leq N}$ and the collision kernels $(B_{ij})_{1\leq i,j\leq N}$.
\end{teor}
 
In order to take advantage of this result, we choose to set our study in the space $\spaceLMR$. Following the path of its proof, it appears that we strongly need the fact that the operator $\boldL$ is self-adjoint in $\spaceLMR$, which is not the case for $\mathbf{\epsL}$. The lack of this property does not allow to directly apply the tools developed in \cite{BriDau} to show the existence of a spectral gap for $\mathbf{\epsL}$. In particular, one of the main issues we encountered is that we were not able to recover the form of $\Ker\mathbf{\epsL}$ in $\spaceLMR$, leaving the corresponding orthogonal projection $\pi_{\mathbf{\epsL}}$ onto $\Ker\mathbf{\epsL}$ undefined. This led us to confront the problem from a different point of view. 

The idea is to recover an estimate for the Dirichlet form of $\mathbf{\epsL}$ in $\spaceLMR$, by finding a way to link the two operators and use what is known about the spectral properties of $\boldL$ on this space. The price we have to pay is the appearence of two positive terms of order $\varepsilon$ in the estimate. The first one is a correction to the spectral gap $\lambda_{\boldL}$, which comes from the choice of working with the orthogonal projection $\pi_{\boldL}$ onto $\Ker\boldL$ in $\spaceLMR$, in place of $\pi_{\mathbf{\epsL}}$. In fact, the Dirichlet form of $\mathbf{\epsL}$ is taken with respect to the global equilibrium $\pmb{\mu}$ instead of a local Maxwellian $\mathbf{\epsM}$, and the correction to $\lambda_{\boldL}$ translates the fact that $\pmb{\mu}$ and $\mathbf{\epsM}$ are close in Euclidean norm, up to a factor of order $\varepsilon$. Moreover, a further consequence of our choice is that $\pi_{\boldL}(\boldf)\notin\Ker\mathbf{\epsL}$, meaning that $\mathbf{\epsL}(\pi_{\boldL}(\boldf))$ is not zero and has to be estimated, producing another contribution of order $\varepsilon$. The second term precisely takes into account this last feature.

\vspace{2mm}
Let us now state the main result that we shall prove.

\begin{teor}\label{MainResult}
Let the collision kernels $B_{ij}$ satisfy assumptions $\mathrm{(H1)-(H4)}$, and let $\lambda_{\boldL}>0$ be the spectral gap in $\spaceLMR$ of the operator $\boldL$ defined by (\ref{GlobalLinearized1}). There exists $C>0$ such that, for all $\varepsilon >0$, the operator $\mathbf{\epsL}$ defined by (\ref{LocalLinearized1}) satisfies the following \textrm{a priori} estimate. For almost every $t\geq 0$ and $x\in\T^3$, and for any $\boldf\in \spaceLMR$, 
\begin{equation}\label{MainEstimate}
\hspace{0.2cm}\big\langle \mathbf{\epsL}(\boldf), \boldf\big\rangle_{\spaceLM}\leq -\left(\lambda_{\boldL}-C\varepsilon\max_{1\leq i\leq N} \Big\{c_{i,\infty}^{1-\delta}|u_i(t,x)|\mathcal{R}_{\varepsilon, i}\Big\}\right)\|\boldf-\pi_{\boldL}(\boldf)\|_{\spaceLvM}^2
\end{equation}

\[
\hspace{0.6cm}+C\varepsilon\max_{1\leq i\leq N} \Big\{c_{i,\infty}^{1-\delta}|u_i(t,x)|\mathcal{R}_{\varepsilon, i}\Big\}\|\pi_{\boldL}(\boldf)\|_{\spaceLvM}^2,
\]
where $\pi_{\boldL}$ is the orthogonal projection onto $\Ker\boldL$ in $\spaceLMR$, and we have defined

\[
\mathcal{R}_{\varepsilon, i}=1+\varepsilon|u_i(t,x)|\exp\left\{\frac{4 m_i}{1-\delta}\varepsilon^2|u_i(t,x)|^2\right\},
\]
for some $\delta\in (0,1)$ independent of $\varepsilon$ to be appropriately chosen.

The constant $C$ is explicit and only depends on $N$, the different masses $(m_i)_{1\leq i\leq N}$, the collision kernels $(B_{ij})_{1\leq i,j\leq N}$ and the parameter $\delta\in (0,1)$. In particular, it does not depend on $\varepsilon$, the number of particles of the different species $(c_{i,\infty})_{1\leq i\leq N}$ and the macroscopic velocities $(u_i)_{1\leq i\leq N}$.
\end{teor}
 This result provides an estimate which is local in space and time, meaning that $(t,x)$ remains fixed. For the sake of simplicity, from now on we shall omit the explicit dependence of the macroscopic velocity $u_i=u_i(t,x)$ on both variables. 

\vspace{1mm}
As already mentioned, the main idea of the proof consists in linking $\mathbf{\epsL}$ to $\boldL$ in order to use the properties known on the linearized Boltzmann operator and apply Theorem \ref{SpectralGapL}. In view of this, we write $\mathbf{\epsL}$ as a sum $\mathbf{\epsL}=\boldL+(\mathbf{\epsL}-\boldL)$ and we separately treat the two addends. The operator $\boldL$ will be straightforwardly controlled using Theorem \ref{SpectralGapL}, producing the contribution $\lambda_{\boldL}$ in the estimate of Theorem \ref{MainResult}, and we shall need to prove that $\mathbf{\epsL}-\boldL$ can be bounded and gives the remaining contributions of order $\varepsilon$. 

In order to show the control on the penalty term, we first recover a bound of order $\varepsilon$ for the differences $\epsM_i-\mu_i$, proving that even if the local Maxwellian distributions (\ref{EpsilonMaxwellian}) constitute a non-equilibrium state for the Boltzmann multi-species operator, they are still very close to the global equilibrium $\pmb{\mu}$. To take advantage of this property, the penalty term will be further splitted into two parts that we shall study independently. The first one will be a mulplicative operator that shall be easily controlled, while the second one will be an integral operator that shall be written under a kernel form (using the Carleman representation) and estimated pointwise with the strategy proposed in \cite{BriDau}. Here, we will encounter a major issue. In fact, we shall need to manage the complete lack of symmetry between cross-interacting species (different atomic mass $m_i$ and macroscopic velocity $\varepsilon u_i$ for each species of the mixture), which will lead to intricate mixing effects in the exponential decay rates of the Maxwellian weights. To solve the difficulty we shall establish new accurate $L^1$ controls in the velocity variable for the kernel operator. This will involve proving new technical lemmata, but will finally allow to close the estimate and recover our result.


\section{Proof of Theorem \ref{MainResult}}

This section is devoted to the proof of our main result. In order to explain in a clear way the procedure, we divide the proof into several steps. We first split the operator $\mathbf{\epsL}$ into the linearized Boltzmann operator $\boldL$ plus a penalization term $\mathbf{\epsL}-\boldL$, and we study their associated Dirichlet forms. We then use Theorem \ref{SpectralGapL} to bound the Dirichlet form of $\boldL$, while we note that the Dirichlet form of $\mathbf{\epsL}-\boldL$ produces four terms. Each of these terms will be bounded separately in a dedicated step, proving that they produce a contribution of order $\varepsilon$. We shall finally collect all the estimates to recover Theorem \ref{MainResult} in the last step of the proof. 


\vspace{2mm}
\noindent \textbf{Step 1 -- Penalization of the operator $\mathbf{\epsL}$.} We consider the Dirichlet form of the operator $\mathbf{\epsL}$ in the space $\spaceLMR$, and we apply a splitting based on a penalization by the operator $\boldL$. More precisely, using Theorem \ref{SpectralGapL} we successively get, for any $\boldf\in\spaceLMR$,

\[
\hspace{-2.7cm}\big\langle \mathbf{\epsL}(\boldf), \boldf\big\rangle_{\spaceLM}=\big\langle \boldL(\boldf), \boldf\big\rangle_{\spaceLM}+\big\langle \mathbf{\epsL}(\boldf)-\boldL(\boldf), \boldf\big\rangle_{\spaceLM}
\]

\[
\hspace{3cm}\leq -\lambda_{\boldL}\|\boldf-\pi_{\boldL}(\boldf)\|_{\spaceLvM}^2+\big\langle \mathbf{\epsL}(\boldf)-\boldL(\boldf), \boldf\big\rangle_{\spaceLM},
\]
Then the penalty term writes
\[
\big\langle \mathbf{\epsL}(\boldf)-\boldL(\boldf), \boldf\big\rangle_{\spaceLM}=\sum_{i,j} \int_{\R^6\times\Sf}B_{ij}(|v-v_*|,\cos\theta) \Delta_{ij}^{\varepsilon}f_i \mu_i^{-1}\dd\sigma\dd v_*\dd v,
\]
with the shorthand notation

\[
\Delta_{ij}^{\varepsilon}:=\big({\epsM_i}^{\prime}-\mu_i^{\prime}\big)f_j^{\prime *}+\big({\epsM_j}^{\prime *}-\mu_j^{\prime *}\big)f_i^{\prime}-\big({\epsM_i}-\mu_i\big)f_j^{*}-\big({\epsM_j}^{*}-\mu_j^{*}\big)f_i.
\]
By simply taking the modulus, we can bound from above $\Delta_{ij}^{\varepsilon}$ in this way

\[
\Delta_{ij}^{\varepsilon}\leq\big|{\epsM_i}^{\prime}-\mu_i^{\prime}\big|\big|f_j^{\prime *}\big|+\big|{\epsM_j}^{\prime *}-\mu_j^{\prime *}\big|\big|f_i^{\prime}\big|+\big|{\epsM_i}-\mu_i\big|\big|f_j^{*}\big|+\big|{\epsM_j}^{*}-\mu_j^{*}\big|\big|f_i\big|,
\]
and to conclude the proof we shall need to control each of the four terms appearing in the inequality. For the sake of simplicity, we will separately consider the different terms, and we set
\newpage
\[
\hspace{-1mm}\mathcal{T}_1^{\varepsilon}(\boldf)=\sum_{i,j} \int_{\R^6\times\Sf}B_{ij}(|v-v_*|,\cos\theta)\big|{\epsM_i}^{\prime}-\mu_i^{\prime}\big|\big|f_j^{\prime *}\big|\big|f_i\big| \mu_i^{-1}\dd\sigma\dd v_*\dd v,
\]
\[
\mathcal{T}_2^{\varepsilon}(\boldf)=\sum_{i,j} \int_{\R^6\times\Sf}B_{ij}(|v-v_*|,\cos\theta) \big|{\epsM_j}^{\prime *}-\mu_j^{\prime *}\big|\big|f_i^{\prime}\big|\big|f_i \big|\mu_i^{-1}\dd\sigma\dd v_*\dd v,
\]
\[
\hspace{-3mm}\mathcal{T}_3^{\varepsilon}(\boldf)=\sum_{i,j} \int_{\R^6\times\Sf}B_{ij}(|v-v_*|,\cos\theta) \big|{\epsM_i}-\mu_i\big|\big|f_j^{*}\big|\big|f_i \big|\mu_i^{-1}\dd\sigma\dd v_*\dd v,
\]
\[
\hspace{-0.9cm}\mathcal{T}_4^{\varepsilon}(\boldf)=\sum_{i,j} \int_{\R^6\times\Sf}B_{ij}(|v-v_*|,\cos\theta) \big|{\epsM_j}^{*}-\mu_j^{*}\big|f_i^2 \mu_i^{-1}\dd\sigma\dd v_*\dd v.
\]

\vspace{2mm}
\noindent \textbf{Step 2 -- Preliminary lemmata.} In order to bound each component, we collect in this step two useful preliminary results. The first one yields an estimate which allows to compare the growth of the local Maxwellians given in (\ref{EpsilonMaxwellian}) with the growth of the global Maxwellians (\ref{GlobalEquilibrium}).

\begin{lemm}\label{Lemma1}
Consider the local Maxwellian densities $\mathbf{\epsM}$ defined in (\ref{EpsilonMaxwellian}) and the global equilibria $\boldM$ introduced in (\ref{GlobalEquilibrium}). For almost every $t\geq 0$, $x\in\T^3$ and for all $\delta\in(0,1)$, there exists an explicit constant $C_{\delta}>0$ (independent of $\varepsilon$, of the number of particles $c_{i,\infty}$ and of the macroscopic velocities $u_i$) such that, for all $\varepsilon >0$ and $1\leq i\leq N$,

\begin{equation}\label{Estimate1}
\big|\epsM_i(v)-\mu_i(v)\big|\leq C_{\delta}\varepsilon c_{i,\infty}^{1-\delta}|u_i|\mathcal{R}_{\varepsilon, i}\mu_i^{\delta}(v),\quad\forall\ v\in\R^3,
\end{equation}
recalling that we have defined for simplicity
\[
\mathcal{R}_{\varepsilon, i}=1+\varepsilon|u_i|\exp\left\{\frac{4 m_i}{1-\delta}\varepsilon^2|u_i|^2\right\}.
\]
\end{lemm}

\begin{proof}
Fix $\delta\in (0,1)$. Noticing that $\epsM_i(v)=\mu_i(v-\varepsilon u_i)$ and using Cauchy-Schwarz inequality, we obtain, for any $v\in\R^3$,

\[
\hspace{-3.5cm}\big|\epsM_i(v)-\mu_i(v)\big|=\left|\int_0^1 m_i \langle v-s\varepsilon u_i,\varepsilon u_i\rangle \mu_i(v-s\varepsilon u_i)\dd s\right|
\] 
\[
\hspace{1.2cm}\leq \varepsilon |u_i|\mu_i^{\delta}(v) \int_0^1 m_i |v-s\varepsilon u_i | \mu_i(v-s\varepsilon u_i)\mu_i^{-\delta}(v)\dd s.
\]
The product of the two Maxwellians inside the integral can be upper-estimated as

\[
\mu_i(v-s\varepsilon u_i)\mu_i^{-\delta}(v)\leq c_{i,\infty}^{1-\delta}\left(\frac{m_i}{2\pi}\right)^{3(1-\delta)/2}\exp\left\{-(1-\delta)\frac{m_i}{2}|v|^2+s\varepsilon m_i|u_i||v|\right\},
\]
Consequently, the previous integral is bounded by
\[
\hspace{-0.2cm}\int_0^1 m_i |v-s\varepsilon u_i | \mu_i(v-s\varepsilon u_i)\mu_i^{-\delta}(v)\dd s\leq m_i c_{i,\infty}^{1-\delta}\left(\frac{m_i}{2\pi}\right)^{\frac{3(1-\delta)}{2}}\Big(|v|+\varepsilon|u_i|\Big)e^{-(1-\delta)\frac{m_i}{2}|v|^2+\varepsilon m_i |u_i||v|}.
\]
Starting from this last inequality, we distinguish two cases:

\begin{enumerate}
\item For $\displaystyle |v|\geq \frac{4 \varepsilon|u_i|}{1-\delta}$, we have
\[
\int_0^1 m_i |v-s\varepsilon u_i | \mu_i(v-s\varepsilon u_i)\mu_i^{-\delta}(v)\dd s\leq m_i \left(\frac{m_i}{2\pi}\right)^{\frac{3(1-\delta)}{2}}\frac{5-\delta}{4}|v|\exp\left\{-(1-\delta)\frac{m_i}{4}|v|^2\right\}
\]
which is clearly upper bounded since $\delta\in (0,1)$.
\item For $\displaystyle |v|< \frac{4 \varepsilon|u_i|}{1-\delta}$, we get
\[
\int_0^1 m_i |v-s\varepsilon u_i | \mu_i(v-s\varepsilon u_i)\mu_i^{-\delta}(v)\dd s\leq m_i \left(\frac{m_i}{2\pi}\right)^{\frac{3(1-\delta)}{2}} \frac{5-\delta}{1-\delta}\varepsilon|u_i|\exp\left\{\frac{4 m_i}{1-\delta}\varepsilon^2|u_i|^2\right\}.
\]
\end{enumerate}
\noindent For all $v\in\R^3$, we finally deduce the bound 
\[
\int_0^1 m_i |v-s\varepsilon u_i | \mu_i(v-s\varepsilon u_i)\mu_i^{-\delta}(v)\dd s\leq C_\delta \left(1+\varepsilon|u_i|\exp\left\{\frac{4 m_i}{1-\delta}\varepsilon^2|u_i|^2\right\}\right)
\]
by simply taking
\[
C_\delta = \max_{1\leq i\leq N}  m_i \left(\frac{m_i}{2\pi}\right)^{\frac{3(1-\delta)}{2}}(5-\delta)\left(\sup_{|v|\in\R^+}\left(\frac{|v|}{4}e^{-(1-\delta)\frac{m_i}{4}|v|^2}\right)+\frac{5-\delta}{1-\delta}\right) >0,
\]
which concludes the proof.
\end{proof}

We also present a second important result that we shall use in the following. In that end, we first introduce an additional property of the linearized Boltzmann operator.

The operator $\boldL$ can be written under the form
\[
\boldL=\mathbf{K}-\pmb{\nu},
\]
where $\mathbf{K}=(K_1,\ldots,K_N)$ is a compact operator \cite{BouGrePavSal} defined, for any $1\leq i\leq N$, by
\begin{equation}\label{OperatorK}
K_i(\boldf)(v)=\sum_{j=1}^N\int_{\R^3\times\Sf}B_{ij}(|v-v_*|,\cos\theta)\Big(\mu_i^{\prime}f_j^{\prime *}+\mu_j^{\prime *}f_i^{\prime}-\mu_i f_j^*\Big)\dd\sigma\dd v_*,\hspace{0.5cm}\forall\ v\in\R^3 ,
\end{equation}
and $\pmb{\nu}=(\nu_1,\ldots ,\nu_N)$ acts like a multiplicative operator, namely, for any $1\leq i\leq N$,
\[
\nu_i(\boldf)(v)=\sum_{j=1}^N \nu_{ij}(v)f_i(v),
\]
\begin{equation}\label{CollisionFrequency}
\nu_{ij}(v)=\int_{\R^3} B_{ij}\big(|v-v_*|,\cos\theta\big)\mu_j(v_*)\dd\sigma\dd v_*,\hspace{0.5cm} \forall\ v\in\R^3.
\end{equation}
It is well-known \cite{Cer, CerIllPul, Vil} that, under our assumptions on the collision kernels $B_{ij}$, each of the so-called collision frequencies $\nu_{ij}$ satisfies the following lower and upper bounds. For all $1\leq i,j\leq N$, there exist two explicit positive constants $\bar{\nu}_{ij}, \tilde{\nu}_{ij}>0$  such that 

\begin{equation}\label{BoundCF}
0< \bar{\nu}_{ij}\leq \nu_{ij}(v)\leq \tilde{\nu}_{ij}\left(1+|v|^{\gamma}\right),\hspace{0.5cm}\forall\ v\in\R^3.
\end{equation}

The lemma we then recall has been proved in \cite[Lemma 5.1]{BriDau} (see also \cite{BouGrePavSal}), and states that the operator $\mathbf{K}$ can be written under a kernel form. The authors actually prove a more general result, showing that the full Boltzmann multi-species operator (and not only its linearized version) has a Carleman representation \cite{Car}.

\begin{lemm}\label{Lemma2}
For any $\boldf=(f_1,\ldots ,f_N)\in\spaceLMR$, the following Carleman representations hold:
\[
\hspace{-8cm}\int_{\R^3\times\Sf}B_{ij}(|v-v_*|,\cos\theta)f_i^{\prime *}f_j^{\prime}\ \dd\sigma\dd v_*= 
\]
\[
C_{ij} \int_{\R^3}\left(\frac{1}{|v-v_*|}\int_{\tilde{E}^{ij}_{v v_*}}\frac{B_{ij}\left(v-V(u,v_*),\frac{v_*-u}{|u-v_*|}\right)}{|u-v_*|}f_i(u)\ \dd \tilde{E}(u)\right)f_j^*\dd v_*,
\]

\vspace{2mm}
\[
\hspace{-8cm}\int_{\R^3\times\Sf}B_{ij}(|v-v_*|,\cos\theta)f_i^{\prime}f_j^{\prime *}\ \dd\sigma\dd v_*= 
\]
\[
C_{ji} \int_{\R^3}\left(\frac{1}{|v-v_*|}\int_{E^{ij}_{v v_*}}\frac{B_{ij}\left(v-V(v_*,u),\frac{u-v_*}{|u-v_*|}\right)}{|u-v_*|}f_j(u)\ \dd E(u)\right)f_i^*\dd v_*,
\]
where we have defined $V(u,v_*)=v_*+m_i m_j^{-1}u-m_i m_j^{-1}v$ and called $C_{ij},C_{ji}>0$ some explicit constants which only depend on the masses $m_i, m_j$. Moreover, we have denoted by $\dd E$ the Lebesgue measure on the hyperplane $E_{v v_*}^{ij}$, orthogonal to $v-v_*$ and passing through 
\[
V_E(v,v_*)=\frac{m_i+m_j}{2 m_j}v-\frac{m_i-m_j}{2m_j}v_*,
\]
and by $\dd \tilde{E}$ the Lebesgue measure on the space $\tilde{E}_{v v_*}^{ij}$ which corresponds to the hyperplane $E_{v v_*}^{ij}$, whenever $m_i=m_j$, and to the sphere of radius
\[
R_{v v_*}=\frac{m_j}{|m_i-m_j|}|v-v_*|
\]
and centred at 
\[
O_{v v_*}=\frac{m_i}{m_i-m_j}v-\frac{m_j}{m_i-m_j}v_*,
\]
whenever $m_i\neq m_j$. 
\end{lemm}

Let us now pass to the actual estimate of the different terms $\mathcal{T}_i^{\varepsilon}$. For the sake of simplicity, we shall call $C^{(i)}$ any constant bounding the corresponding term $\mathcal{T}_i^{\varepsilon}$ independently of $\varepsilon$ and $u_i$, but each of these constants will be explicit in the sense that all steps in our proofs are constructive and explicit computations of the constants could be extracted from them.

\vspace{2mm}
\noindent \textbf{Step 3 -- Estimation of \hspace{0.5mm}$\bm{\mathcal{T}_4^{\varepsilon}}$.} Applying Lemma \ref{Lemma1} for any fixed $\delta\in (0,1)$ and using the assumptions made on the collisional kernels, the term can be estimated as
\[
\hspace{-1.7cm}\mathcal{T}_4^{\varepsilon}(\boldf)=\sum_{i,j} \int_{\R^6\times\Sf}B_{ij}(|v-v_*|,\cos\theta) \big|{\epsM_j}^{*}-\mu_j^{*}\big|f_i^2 \mu_i^{-1}\dd\sigma\dd v_*\dd v
\]
\[
\hspace{-0.2cm}\leq \varepsilon C^{(4)}\sum_{i,j} c_{j,\infty}^{1-\delta}|u_j|\mathcal{R}_{\varepsilon, j}\int_{\R^3}\left(\int_{\R^3}|v-v_*|^{\gamma}{\mu_j^*}^{\delta}\dd v_*\right){f_i}^2\mu_i^{-1}\dd v
\]
\[
\hspace{1.4cm}\leq \varepsilon C^{(4)} \max_{1\leq i\leq N} \Big\{c_{i,\infty}^{1-\delta}|u_i|\mathcal{R}_{\varepsilon, i}\Big\}\sum_{i,j}\int_{\R^3}\left(\int_{\R^3}|v-v_*|^{\gamma}{\mu_j^*}^{\delta}\dd v_*\right){f_i}^2\mu_i^{-1}\dd v.
\]
Moreover, we notice that for all $1\leq j\leq N$, the integral term 
\[
\nu_{j}^{(\delta)}= \int_{\R^3}|v-v_*|^{\gamma}{\mu_j^*}^{\delta}\dd v_*
\]
has the same form of the collision frequencies $\nu_{ij}$ and it consequently satisfies an estimate similar to (\ref{BoundCF}), namely there exist two strictly positive explicit constants $\bar{\nu}_{j}^{(\delta)}, \tilde{\nu}_{j}^{(\delta)}>0$  such that 
\[
0<\bar{\nu}_{j}^{(\delta)}\leq \nu_{j}^{(\delta)}(v)\leq \tilde{\nu}_{j}^{(\delta)}\left(1+|v|^{\gamma}\right), \quad\forall v\in\R^3.
\]
Using the fact that $\left(1+|v|^{\gamma}\right)\sim\langle v\rangle^{\gamma}$, we finally deduce
\[
\mathcal{T}_4^{\varepsilon}\leq \varepsilon C^{(4)} \max_{1\leq i\leq N} \Big\{c_{i,\infty}^{1-\delta}|u_i|\mathcal{R}_{\varepsilon, i}\Big\} \sum_{i=1}^N\int_{\R^3}\langle v\rangle^{\gamma}f_i^2 \mu_i^{-1}\dd v
\]
\begin{equation}\label{EstimationK4}
\hspace{0.3cm}\leq \varepsilon C^{(4)} \max_{1\leq i\leq N} \Big\{c_{i,\infty}^{1-\delta}|u_i|\mathcal{R}_{\varepsilon, i}\Big\} \big\|\boldf\big\|_{\spaceLvM}^2.
\end{equation}

\vspace{2mm}
\noindent \textbf{Step 4 -- Estimation of \hspace{0.5mm}$\bm{\mathcal{T}_3^{\varepsilon}}$.} Similarly, we can rewrite $\mathcal{T}_3^{\varepsilon}$ in term of $\mathbf{h}=(h_1,\ldots ,h_N)$ with $h_i =f_i \mu_i^{-1/2}$ for any $1\leq i\leq N$, and obtain

\[
\hspace{-1.4cm}\mathcal{T}_3^{\varepsilon}(\boldf)=\sum_{i,j} \int_{\R^6\times\Sf}B_{ij}(|v-v_*|,\cos\theta) \big|{\epsM_i}-\mu_i\big|\mu_i^{-1/2}{\mu_j^*}^{1/2}\big|h_j^{*}\big|\big|h_i\big|\dd\sigma\dd v_*\dd v
\]
\[
\hspace{0.1cm}\leq \varepsilon C^{(3)}\max_{1\leq i\leq N} \Big\{c_{i,\infty}^{1-\delta}|u_i|\mathcal{R}_{\varepsilon, i}\Big\}\sum_{i,j}\int_{\R^6} |v-v_*|^{\gamma} \mu_i^{\frac{2\delta-1}{2}}{\mu_j^*}^{1/2}\big|h_j^{*}\big|\big|h_i\big|\dd v_*\dd v,
\]
where this time we have used Lemma \ref{Lemma1} with $\delta\in(1/2,1)$. Applying Young's inequality and Fubini's theorem allows to conclude that

\[
\mathcal{T}_3^{\varepsilon}(\boldf)\leq \varepsilon C^{(3)}\max_{1\leq i\leq N} \Big\{c_{i,\infty}^{1-\delta}|u_i|\mathcal{R}_{\varepsilon, i}\Big\}\sum_{i,j}\left(\int_{\R^6} |v-v_*|^{\gamma} \mu_i^{2\delta -1}{h_j^*}^2\dd v\dd v_* + \int_{\R^6} |v-v_*|^{\gamma} \mu_j^* h_i^2\dd v_*\dd v\right)
\]
\[
\hspace{-0.5cm}\leq \varepsilon C^{(3)}\max_{1\leq i\leq N} \Big\{c_{i,\infty}^{1-\delta}|u_i|\mathcal{R}_{\varepsilon, i}\Big\}\left(\sum_{j=1}^N\int_{\R^3} \langle v_* \rangle^{\gamma}{f_j^*}^2 {\mu_j^*}^{-1}\dd v_* + \sum_{i=1}^N\int_{\R^3} \langle v \rangle^{\gamma} f_i^2\mu_i^{-1}\dd v\right)
\]

\vspace{3mm}
\begin{equation}\label{EstimationK3}
\hspace{-5.9cm}\leq \varepsilon C^{(3)} \max_{1\leq i\leq N} \Big\{c_{i,\infty}^{1-\delta}|u_i|\mathcal{R}_{\varepsilon, i}\Big\} \big\|\boldf\big\|_{\spaceLvM}^2.
\end{equation}

\vspace{2mm}
\noindent \textbf{Step 5 -- Estimation of \hspace{0.5mm}$\bm{\mathcal{T}_2^{\varepsilon}}$.} In order to estimate $\mathcal{T}_2^{\varepsilon}$, we first apply Lemma \ref{Lemma1}, like in the case of $\mathcal{T}_4^{\varepsilon}$, to bound the difference between the local and the global Maxwellians. We have 

\[
\hspace{-3.3cm}\mathcal{T}_2^{\varepsilon}(\boldf)=\sum_{i,j} \int_{\R^6\times\Sf}B_{ij}(|v-v_*|,\cos\theta) \big|{\epsM_j}^{\prime *}-\mu_j^{\prime *}\big|\big|f_i^{\prime}\big|\big|f_i\big| \mu_i^{-1}\dd\sigma\dd v_*\dd v
\]
\[
\hspace{1cm}\leq \varepsilon C^{(2)}\max_{1\leq i\leq N} \Big\{c_{i,\infty}^{1-\delta}|u_i|\mathcal{R}_{\varepsilon, i}\Big\}\sum_{i,j}\int_{\R^6\times\Sf}B_{ij}(|v-v_*|,\cos\theta) {\mu_j^{\prime *}}^{\delta}\big|f_i^{\prime}\big|\big|f_i \big|\mu_i^{-1}\dd\sigma\dd v_*\dd v,
\]
for any $\delta\in (0,1)$. Next, we write the previous integral in its kernel form using Lemma \ref{Lemma2}. That allows to get, for any $1\leq i,j\leq N$,

\[
\hspace{-7.3cm}\int_{\R^6\times\Sf}B_{ij}(|v-v_*|,\cos\theta) {\mu_j^{\prime *}}^{\delta}\big|f_i^{\prime}\big|\big|f_i\big| \mu_i^{-1}\ \dd\sigma\dd v_*\dd v
\]
\[
\hspace{1.3cm}=C_{ji} \int_{\R^6}\left(\frac{1}{|v-v_*|}\int_{E^{ij}_{v v_*}}\frac{B_{ij}\left(v-V(v_*,u),\frac{u-v_*}{|u-v_*|}\right)}{|u-v_*|}{\mu_j}^{\delta}(u)\ \dd E(u)\right)\big|f_i^*\big|\big|f_i\big| \mu_i^{-1}\dd v_*\dd v
\]
\[
\hspace{1.6cm}=C_{ji} \int_{\R^6}\left(\frac{1}{|v-v_*|}\int_{E^{ij}_{v v_*}}\frac{B_{ij}\left(v-V(v_*,u),\frac{u-v_*}{|u-v_*|}\right)}{|u-v_*|}{\mu_j}^{\delta}(u)\ \dd E(u)\right)\sqrt{\frac{\mu_i^*}{\mu_i}}\big|h_i^*\big|\big|h_i\big|\dd v_*\dd v.
\]

\vspace{1mm}
In order to estimate this integral, we use a result presented in \cite[Lemma 5.1]{BriDau}, in which the authors analyze the particular structure of the kernel form of the operator $\mathbf{K}$ and they deduce explicit pointwise bounds on it. The strategy involves a careful treatement of the different exponential decay rates appearing in the Maxwellian weights. This feature is specific to the multi-species framework and is a direct consequence of the asymmetry in the cross-interactions between species of the mixture, where the different atomic masses lead to these intricate mixing effects. Their result is proved in the case of global equilibria $\mu_i$, but it also straightforwardly applies to any function of the form $\mu_i^{\delta}$, with $\delta >0$. We omit the proof, since it follows exactly the same steps as in their work. 

\begin{lemm}\label{Lemma3}
For any $\delta\in (0,1)$ and $1\leq i,j\leq N$, let the operator $I_{ji}^{\delta}$ be defined for $v,v_*\in\R^3$ by
\[
I_{ji}^{\delta}(v,v_*)=\frac{C_{ji}}{|v-v_*|}\int_{E^{ij}_{v v_*}}\frac{B_{ij}\left(v-V(v_*,u),\frac{u-v_*}{|u-v_*|}\right)}{|u-v_*|}{\mu_j}^{\delta}(u)\ \dd E(u).
\]
Then, there exist two explicit constants $m, C_{ji}^{\delta}>0$ such that
\[
I_{ji}^{\delta}(v,v_*)\leq C_{ji}^{\delta}|v-v_*|^{\gamma -2}e^{-\delta m|v-v_*|^2-\delta m\frac{\left||v|^2-|v_*|^2\right|^2}{|v-v_*|^2}}\left(\frac{\mu_i(v)}{\mu_i(v_*)}\right)^{\delta/2},\quad \forall\ v,v_*\in\R^3.
\]
The constant $m$ only depends on the masses $m_i,m_j$, while the constant $C_{ji}^{\delta}$ only depends on $\delta$, the masses $m_i,m_j$ and the collisional kernel $B_{ij}$.
\end{lemm}

 We prove here that this estimate allows us to recover an upper bound for $I_{ij}^{\delta}$ in a weighted $L_v^1$-norm. We follow closely the proof presented in \cite[Lemma 7]{Guo}.

\begin{lemm}\label{Lemma4}
There exists $\bar{\delta}\in (0,1)$ such that, for each $\delta\in (\bar{\delta},1)$, there exists a constant $C_{ji}(\delta)>0$ such that
\[
\hspace{-1cm}\int_{\R^3} I_{ji}^{\delta}(v,v_*)\sqrt{\frac{\mu_i(v_*)}{\mu_i(v)}}\ \dd v_*\leq C_{ji}(\delta),\hspace{0.5cm}\forall\ v\in\R^3,
\]

\[
\hspace{-0.7cm}\int_{\R^3} I_{ji}^{\delta}(v,v_*)\sqrt{\frac{\mu_i(v_*)}{\mu_i(v)}}\ \dd v\leq C_{ji}(\delta),\quad\forall\ v_*\in\R^3.
\]
The constant $C_{ji}(\delta)$ is explicit and only depends on the parameter $\delta\in (\bar{\delta},1)$, the masses $m_i,m_j$ and the collisional kernel $B_{ij}$.
\end{lemm}

\begin{proof}
We shall only prove the first inequality, since the second one can be showed in a similar way. Using the previous lemma, the term inside the integral can be bounded for all $v,v_*\in\R^3$ by

\[
I_{ji}^{\delta}(v,v_*)\sqrt{\frac{\mu_i(v_*)}{\mu_i(v)}}\leq C_{ji}^{\delta}|v-v_*|^{\gamma -2}e^{-\delta m|v-v_*|^2-\delta m\frac{\left||v|^2-|v_*|^2\right|^2}{|v-v_*|^2}}\left(\frac{\mu_i(v_*)}{\mu_i(v)}\right)^{\frac{1-\delta}{2}}.
\]
After renaming $\eta=v-v_*$ and $v_*=v-\eta$ in the integral (for the second inequality we should set $\eta=v_*-v$ and $v=v_*-\eta$), direct algebraic manipulations on the total exponent in the last inequality yield

\[
\hspace{-1.1cm}-\delta m|\eta|^2-\delta m\frac{\left||\eta|^2-2\langle \eta,v\rangle\right|^2}{|\eta|^2}-(1-\delta)\frac{m_i}{4}\left(|v-\eta|^2-|v|^2\right)
\]
\[
\hspace{-0.8cm}=-2\delta m|\eta|^2-4\delta m\frac{\langle \eta,v\rangle^2}{|\eta|^2}+4\delta m\langle \eta,v\rangle-(1-\delta)\frac{m_i}{4}\left(|\eta|^2-2\langle \eta,v\rangle\right)
\]
\[
=\left(-2\delta m-(1-\delta)\frac{m_i}{4}\right)|\eta|^2+\left(4\delta m+(1-\delta)\frac{m_i}{2}\right)\langle \eta,v\rangle -4\delta m\frac{\langle \eta,v\rangle^2}{|\eta|^2}.
\]
The discriminant of the above quadratic form in $|\eta|$ and $ \frac{\langle \eta,v\rangle}{|\eta|}$ can be computed and gives

\[
\Delta({\delta}) = \left(4\delta m+(1-\delta)\frac{m_i}{2}\right)^2-8\delta m\left(4\delta m+(1-\delta)\frac{m_i}{2}\right)
\]
\[
\hspace{-0.6cm}=\frac{1}{2}\left(4\delta m+(1-\delta)\frac{m_i}{2}\right)\big(1-\delta(8m+1)\big).
\]
If we now choose $\bar{\delta}=\frac{1}{8m+1}$, we clearly have that for all $\delta\in(\bar{\delta},1)$, $\Delta(\delta)<0$ which implies that the quadratic form is negative definite. Consequently, for all $\delta\in(\bar{\delta},1)$, there exists a constant $C(\delta)>0$ such that 

\[
\hspace{-0.5cm}-\delta m|\eta|^2-\delta m\frac{\left||\eta|^2-2\langle \eta,v\rangle\right|^2}{|\eta|^2}-(1-\delta)\frac{m_i}{4}\left(|v-\eta|^2-|v|^2\right)\leq -C(\delta) \left(|\eta|^2+\frac{\langle\eta, v\rangle^2}{|\eta|^2}\right)
\]
\[
\hspace{7.4cm}\leq -C(\delta) |\eta|^2.
\]
Using this last inequality, for each $\delta\in(\bar{\delta},1)$ we deduce that

\[
\hspace{-1cm}\int_{\R^3} I_{ji}^{\delta}(v,v_*)\sqrt{\frac{\mu_i(v_*)}{\mu_i(v)}}\dd v_* \leq C_{ji}(\delta)\int_{\R^3}|\eta|^{\gamma -2}e^{-\delta m|\eta|^2-\delta m\frac{\left||\eta|^2-2\langle \eta,v\rangle\right|^2}{|\eta|^2}-(1-\delta)\frac{m_i}{4}\left(|v-\eta|^2-|v|^2\right)}\dd \eta
\]
\[
\hspace{-2.2cm}\leq C_{ji}(\delta)\int_{\R^3}|\eta|^{\gamma -2}e^{-C(\delta)|\eta|^2}\dd \eta
\]

\[
\hspace{-3.7cm} \leq C_{ji}(\delta),\quad\forall\ v\in\R^3, 
\]
since the integral is finite and in particular bounded by a constant only depending on $\delta$. This concludes the proof.

\end{proof}

 Thanks to these results, we can now bound the term $\mathcal{T}_2^{\varepsilon}$. For any $\delta\in(\bar{\delta}_1,1)$ with $\bar{\delta}_1=\frac{1}{8m+1}$, applying again Young's inequality, Fubini's theorem, and using the straightforward bound $C_{ji}(\delta)\leq C_{ij}(\delta)\big(1+|v|^{\gamma}\big)$ for any $v\in\R^3$, we have

\[
\hspace{-4cm}\mathcal{T}_2^{\varepsilon}(\boldf)\leq \varepsilon C^{(2)} \max_{1\leq i\leq N} \Big\{c_{i,\infty}^{1-\delta}|u_i|\mathcal{R}_{\varepsilon, i}\Big\}\sum_{i,j}\int_{\R^6}I_{ji}^{\delta}(v,v_*)\sqrt{\frac{\mu_i^*}{\mu_i}}\big|h_i^*\big||h_i| \dd v_*\dd v
\]
\[
\hspace{0.5cm}\leq \varepsilon C^{(2)} \max_{1\leq i\leq N} \Big\{c_{i,\infty}^{1-\delta}|u_i|\mathcal{R}_{\varepsilon, i}\Big\}\sum_{i,j}\left(\int_{\R^6}I_{ji}^{\delta}(v,v_*)\sqrt{\frac{\mu_i^*}{\mu_i}}{h_i^*}^2 \dd v\dd v_* + \int_{\R^6}I_{ji}^{\delta}(v,v_*)\sqrt{\frac{\mu_i^*}{\mu_i}}h_i^2 \dd v_*\dd v\right)
\]
\[
\hspace{-0.7cm}\leq \varepsilon C^{(2)} \max_{1\leq i\leq N} \Big\{c_{i,\infty}^{1-\delta}|u_i|\mathcal{R}_{\varepsilon, i}\Big\}\sum_{i,j} C_{ji}(\delta)\left(\int_{\R^3}\big(1+|v_*|^{\gamma}\big){h_i^*}^2 \dd v_* + \int_{\R^3}\big(1+|v|^{\gamma}\big)h_i^2 \dd v\right)
\]

\begin{equation}\label{EstimationK2}
\hspace{-7cm}\leq\varepsilon C^{(2)} \max_{1\leq i\leq N} \Big\{c_{i,\infty}^{1-\delta}|u_i|\mathcal{R}_{\varepsilon, i}\Big\} \big\|\boldf\big\|_{\spaceLvM}^2.
\end{equation}

\vspace{2mm}
\noindent \textbf{Step 6 -- Estimation of \hspace{0.5mm}$\bm{\mathcal{T}_1^{\varepsilon}}$.} The estimate for the last term $\mathcal{T}_1^{\varepsilon}$ is similar to the one for $\mathcal{T}_2^{\varepsilon}$, but more tricky. In fact, we shall deal with different masses $m_i,m_j$, which lead to an asymmetry in the Maxwellian weights inside the integral, namely the factor $\sqrt{\frac{\mu_j(v_*)}{\mu_i(v)}}$ appearing in this case. As a consequence, a careful treatment of these weights is required to compare the different exponential decay rates. 

\vspace{1mm}
 We first apply Lemma \ref{Lemma1} to obtain, for all $\delta\in(0,1)$
\[
\hspace{-2.1cm}\mathcal{T}_1^{\varepsilon}(\boldf)=\sum_{i,j} \int_{\R^6\times\Sf}B_{ij}(|v-v_*|,\cos\theta)\big|{\epsM_i}^{\prime}-\mu_i^{\prime}\big|\big|f_j^{\prime *}\big|\big|f_i\big| \mu_i^{-1}\dd\sigma\dd v_*\dd v
\]
\[
\hspace{1.1cm}\leq\varepsilon C^{(1)} \max_{1\leq i\leq N} \Big\{c_{i,\infty}^{1-\delta}|u_i|\mathcal{R}_{\varepsilon, i}\Big\}\sum_{i,j}\int_{\R^6\times\Sf}B_{ij}(|v-v_*|,\cos\theta){\mu_i^{\prime}}^{\delta}\big|f_j^{\prime *}\big| \big|f_i \big|\mu_i^{-1}\dd \sigma\dd v_*\dd v,
\]
with $\delta$ actually ranging in a subinterval of $(0,1)$ as before, to be appropriately chosen later. Thanks to Lemma \ref{Lemma2}, we rewrite the integral under its kernel form as 

\[
\hspace{-7cm}\int_{\R^6\times\Sf}B_{ij}(|v-v_*|,\cos\theta){\mu_i^{\prime}}^{\delta}\big|f_j^{\prime *}\big| \big|f_i\big| \mu_i^{-1}\dd \sigma\dd v_*\dd v
\] 
\[
\hspace{1.8cm}=C_{ij} \int_{\R^6}\left(\frac{1}{|v-v_*|}\int_{\tilde{E}^{ij}_{v v_*}}\frac{B_{ij}\left(v-V(u,v_*),\frac{v_*-u}{|u-v_*|}\right)}{|u-v_*|}{\mu_i}^{\delta}(u)\dd \tilde{E}(u)\right)\big|f_j^*\big|\big|f_i\big| \mu_i^{-1}\dd v_*\dd v
\]
\[
\hspace{1.9cm}=C_{ij} \int_{\R^6}\left(\frac{1}{|v-v_*|}\int_{\tilde{E}^{ij}_{v v_*}}\frac{B_{ij}\left(v-V(u,v_*),\frac{v_*-u}{|u-v_*|}\right)}{|u-v_*|}{\mu_i}^{\delta}(u)\dd \tilde{E}(u)\right)\sqrt{\frac{\mu_j^*}{\mu_i}}\big|h_j^*\big|\big|h_i\big|\dd v_*\dd v.
\]
In the case $m_i=m_j$, the computations made for $\mathcal{T}_2^{\varepsilon}$ apply in the same way, with the kernel operator satisfying the properties proved in the two lemmata of Step 5. We therefore suppose, from now on, that $m_i\neq m_j$. In this case, the equivalent result of Lemma $\ref{Lemma3}$ is the following. 

\begin{lemm}\label{Lemma5}
For any $\delta\in (0,1)$ and $1\leq i,j\leq N$, let the operator $I_{ij}^{\delta}$ be defined for $v,v_*\in\R^3$ by
\[
I_{ij}^{\delta}(v,v_*)=\frac{C_{ij}}{|v-v_*|}\int_{\tilde{E}^{ij}_{v v_*}}\frac{B_{ij}\left(v-V(u,v_*),\frac{v_*-u}{|u-v_*|}\right)}{|u-v_*|}{\mu_i}^{\delta}(u)\ \dd \tilde{E}(u).
\]
Then, there exist two explicit constants $\lambda, C_{ij}^{\delta}>0$, only depending on $\delta$, the masses $m_i,m_j$ and the collisional kernel $B_{ij}$, such that
\[
I_{ij}^{\delta}(v,v_*)\leq C_{ij}^{\delta}|v-v_*|^{\gamma }e^{-\delta \lambda|v-v_*|^2-\delta \lambda\frac{\left||v|^2-|v_*|^2\right|^2}{|v-v_*|^2}-\delta\frac{m_i+m_j}{16}\big|V^{\perp}\big|^2}\left(\frac{\mu_i(v)}{\mu_j(v_*)}\right)^{\delta/2},\quad \forall\ v,v_*\in\R^3,
\]
where we have defined $v+v_*=V^{\parallel}+V^{\perp}$, with $V^{\parallel}$ being the projection onto $\mathrm{Span}(v-v_*)$ and $V^{\perp}$ orthogonal to $\mathrm{Span}(v-v_*)$.
\end{lemm}

 The proof still follows exactly the one presented in \cite[Lemma 5.1]{BriDau} (see also the one presented in \cite{BouGrePavSal}), therefore it is omitted. On the other hand, we will see that the proof of the next result is a bit more intricate than the one of Lemma \ref{Lemma4}.

\begin{lemm}\label{Lemma6}
There exists $\bar{\delta}\in (0,1)$ such that, for each $\delta\in (\bar{\delta},1)$, there exists a constant $C_{ij}(\delta)>0$ such that
\begin{align}
\hspace{-0.5cm}\int_{\R^3} I_{ij}^{\delta}(v,v_*)\sqrt{\frac{\mu_j(v_*)}{\mu_i(v)}}\dd v_*\leq C_{ij}(\delta),\quad\forall\ v\in\R^3,\label{Integral1}\\
\int_{\R^3} I_{ij}^{\delta}(v,v_*)\sqrt{\frac{\mu_j(v_*)}{\mu_i(v)}}\dd v\leq C_{ij}(\delta),\quad\forall\ v_*\in\R^3\label{Integral2}.
\end{align}
The constant $C_{ij}(\delta)$ is explicit and only depends on the parameter $\delta\in (\bar{\delta},1)$, the masses $m_i,m_j$ and the collisional kernel $B_{ij}$.
\end{lemm}

\begin{proof}
We begin by studying the term inside the integral of inequality (\ref{Integral1}). Using Lemma \ref{Lemma5}, we easily deduce that for all $v,v_*\in\R^3$

\begin{equation}\label{TotalExponent}
I_{ij}^{\delta}(v,v_*)\sqrt{\frac{\mu_j(v_*)}{\mu_i(v)}}\leq C_{ij}^{\delta}|v-v_*|^{\gamma }e^{-\delta \lambda|v-v_*|^2-\delta \lambda\frac{\left||v|^2-|v_*|^2\right|^2}{|v-v_*|^2}-\delta\frac{m_i+m_j}{16}\big|V^{\perp}\big|^2}\left(\frac{\mu_j(v_*)}{\mu_i(v)}\right)^{\frac{1-\delta}{2}},
\end{equation}
recalling that we have defined $v+v_*=V^{\parallel}+V^{\perp}$, with $V^{\parallel}$ being the projection onto $\textrm{Span}(v-v_*)$ and $V^{\perp}$ orthogonal to $\textrm{Span}(v-v_*)$.

\vspace{1mm}
 We want to get rid of the asymmetry between the masses $m_i$ and $m_j$ in the Maxwellian weight $\left(\frac{\mu_j(v_*)}{\mu_i(v)}\right)^{\frac{1-\delta}{2}}$. Employing the decomposition $v+v_*=V^{\parallel}+V^{\perp}$, the following equality holds for any $a,b\in\R$:

\vspace{1mm}
\[
\hspace{-4cm}|av-b v_*|^2=\left|\frac{a-b}{2}(v+v_*)+\frac{a+b}{2}(v-v_*)\right|^2
\]
\[
\hspace{-0.9cm}=\frac{(a-b)^2}{4}\big|V^{\perp}\big|^2+\frac{(a-b)^2}{4}\frac{\left||v|^2-|v_*|^2\right|^2}{|v-v_*|^2}
\]
\[
\hspace{1.2cm}+\frac{(a-b)(a+b)}{2}\big(|v|^2-|v_*|^2\big)+\frac{(a+b)^2}{4}|v-v_*|^2,
\]
since by definition

\[
\big|V^{\parallel}\big|^2=\frac{\langle v+v_*,v-v_*\rangle^2}{|v-v_*|^2}=\frac{\left||v|^2-|v_*|^2\right|^2}{|v-v_*|^2}.
\]
\vspace{1mm}
 Then, taking $(a,b)=(1,0)$ and then $(a,b)=(0,1)$, we immediately get

\[
\hspace{-3cm}m_i|v|^2-m_j|v_*|^2=\frac{m_i-m_j}{4}\big|V^{\perp}\big|^2+\frac{m_i-m_j}{4}\frac{\left||v|^2-|v_*|^2\right|^2}{|v-v_*|^2}
\]
\[
\hspace{1cm}+\frac{m_i+m_j}{2}\big(|v|^2-|v_*|^2\big)+\frac{m_i-m_j}{4}|v-v_*|^2
\]
We can now apply this last equality in the study of the total exponent of (\ref{TotalExponent}). We first estimate

\[
\hspace{-1cm}-\delta \lambda|v-v_*|^2-\delta \lambda\frac{\left||v|^2-|v_*|^2\right|^2}{|v-v_*|^2}-\delta\frac{m_i+m_j}{16}\big|V^{\perp}\big|^2+\frac{1-\delta}{4}\big(m_i|v|^2-m_j|v_*|^2\big)
\]

\[
\hspace{-0.5cm}=\left(-\delta\lambda+(1-\delta)\frac{m_i-m_j}{16}\right)\left(|v-v_*|^2+\frac{\left||v|^2-|v_*|^2\right|^2}{|v-v_*|^2}\right)+(1-\delta)\frac{m_i+m_j}{8}\big(|v|^2-|v_*|^2\big)
\]
\[
\hspace{-9.4cm}-\frac{m_i}{8}\left(\delta-\frac{m_i-m_j}{2m_i}\right)\big|V^{\perp}\big|^2
\]
\begin{equation}\label{Condition1}
\hspace{0.3cm}\leq \left(-\delta\lambda+(1-\delta)\frac{m_i-m_j}{16}\right)\left(|v-v_*|^2+\frac{\left||v|^2-|v_*|^2\right|^2}{|v-v_*|^2}\right)+(1-\delta)\frac{m_i+m_j}{8}\big(|v|^2-|v_*|^2\big),
\end{equation}
whenever $\delta>\max\left\{0,\frac{m_i-m_j}{2 m_i}\right\}$.

\vspace{1mm}
 As previously done, we inject $\eta=v-v_*$ and $v_*=v-\eta$ in the right-hand side of (\ref{Condition1}) and we obtain that the total exponent writes

\[
\hspace{-0.5cm}\left(-\delta\lambda+(1-\delta)\frac{m_i-m_j}{16}\right)\left(|\eta|^2+\frac{\left||\eta|^2-2\langle \eta,v\rangle\right|^2}{|\eta|^2}\right)-(1-\delta)\frac{m_i+m_j}{8}\big(|v-\eta|^2-|v|^2\big)
\]
\begin{equation}\label{QuadraticForm}
=\bigg(-2\delta\lambda-(1-\delta)\frac{m_j}{4}\bigg)|\eta|^2+\bigg(4\delta\lambda+(1-\delta)\frac{m_j}{2}\bigg)\langle\eta,v\rangle +\left(-4\delta\lambda+(1-\delta)\frac{m_i-m_j}{4}\right)\frac{\langle\eta,v\rangle^2}{|\eta|^2}.
\end{equation}
Computing the discriminant of the above quadratic form, we get
\[
\Delta(\delta)=\bigg(4\delta\lambda+(1-\delta)\frac{m_j}{2}\bigg)^2+\bigg(4\delta\lambda+(1-\delta)\frac{m_j}{2}\bigg)\left(-8\delta\lambda+(1-\delta)\frac{m_i-m_j}{2}\right)
\]
\[
\hspace{-3.7cm}=\frac{m_i}{2}\bigg(4\delta\lambda+(1-\delta)\frac{m_j}{2}\bigg)\Bigg(1-\delta\bigg(8\frac{\lambda}{m_i}+1\bigg)\Bigg).
\]
Then we immediately check that if we choose $\bar{\delta}=\max\left\{0,\frac{m_i-m_j}{2m_i},\frac{m_i}{8\lambda+m_i}\right\}$, inequality (\ref{Condition1}) and $\Delta(\delta)<0$ are both satisfied for any $\delta\in(\bar{\delta},1)$. This ensures that, for any $\bar{\delta}<\delta<1$, the quadratic form (\ref{QuadraticForm}) remains negative for any $v\in\R^3$. Consequently, we get that

\[
\hspace{-2cm}-\delta \lambda|\eta|^2-\delta \lambda\frac{\left||\eta|^2-2\langle\eta,v\rangle\right|^2}{|\eta|^2}-\delta\frac{m_i+m_j}{16}\big|V^{\perp}\big|^2-\frac{1-\delta}{4}\big(m_j|v-\eta|^2-m_i|v|^2\big)
\]
\[
\hspace{4.8cm}\leq -C(\delta) \left(|\eta|^2+\frac{\langle\eta, v\rangle^2}{|\eta|^2}\right)\leq -C(\delta) |\eta|^2,
\]
for some explicit constant $C(\delta)>0$ only depending on $\delta$ and the masses $m_i,m_j$.

\vspace{1mm}
 Proceeding as in the proof of Lemma \ref{Lemma4}, for each $\delta\in(\bar{\delta},1)$ the integral in (\ref{Integral1}) can then be bounded as

\[
\hspace{-10.5cm}\int_{\R^3} I_{ij}^{\delta}(v,v_*)\sqrt{\frac{\mu_j(v_*)}{\mu_i(v)}}\dd v_* 
\]
\[
\hspace{2cm}\leq C_{ij}(\delta)\int_{\R^3}|\eta|^{\gamma}e^{-\delta \lambda|\eta|^2-\delta \lambda\frac{\left||\eta|^2-2\langle\eta,v\rangle\right|^2}{|\eta|^2}-\delta\frac{m_i+m_j}{16}\big|V^{\perp}\big|^2-\frac{1-\delta}{4}\big(m_j|v-\eta|^2-m_i|v|^2\big)}\dd \eta
\]
\[
\hspace{-5.2cm}\leq C_{ij}(\delta)\int_{\R^3}|\eta|^{\gamma -2}e^{-C(\delta)|\eta|^2}\dd \eta.
\]

\[
\hspace{-8.8cm}\leq C_{ij}(\delta).
\]

\vspace{1mm}
 Inequality (\ref{Integral2}) can be proved in the same way, by setting this time $\eta =v_*-v$ and using $v=v_*-\eta$. This ends the proof.
\end{proof}

 Young's inequality and Fubini's theorem allow to conclude that, taking $\delta\in (\bar{\delta}_2,1)$ with $ \bar{\delta}_2=\max_{1\leq i,j\leq N}\left\{0,\frac{m_i-m_j}{2m_i},\frac{m_i}{8\lambda+m_i}\right\}$, the term $\mathcal{T}_1^{\varepsilon}$ is bounded by

\vspace{2mm}
\[
\hspace{-4cm}\mathcal{T}_1^{\varepsilon}(\boldf)\leq \varepsilon C^{(1)} \max_{1\leq i\leq N} \Big\{c_{i,\infty}^{1-\delta}|u_i|\mathcal{R}_{\varepsilon, i}\Big\}\sum_{i,j}\int_{\R^6}I_{ij}^{\delta}(v,v_*)\sqrt{\frac{\mu_j^*}{\mu_i}}\big|h_j^*\big|\big|h_i\big| \dd v_*\dd v
\]
\[
\hspace{0.5cm}\leq \varepsilon C^{(1)} \max_{1\leq i\leq N} \Big\{c_{i,\infty}^{1-\delta}|u_i|\mathcal{R}_{\varepsilon, i}\Big\}\sum_{i,j}\left(\int_{\R^6}I_{ij}^{\delta}(v,v_*)\sqrt{\frac{\mu_j^*}{\mu_i}}{h_j^*}^2 \dd v\dd v_* + \int_{\R^6}I_{ij}^{\delta}(v,v_*)\sqrt{\frac{\mu_j^*}{\mu_i}}h_i^2 \dd v_*\dd v\right)
\]
\[
\hspace{-1cm}\leq \varepsilon C^{(1)} \max_{1\leq i\leq N} \Big\{c_{i,\infty}^{1-\delta}|u_i|\mathcal{R}_{\varepsilon, i}\Big\}\left(\sum_{j=1}^N \int_{\R^3}\big(1+|v_*|^{\gamma}\big){h_j^*}^2 \dd v_* +\sum_{i=1}^N  \int_{\R^3}\big(1+|v|^{\gamma}\big)h_i^2 \dd v\right)
\]

\vspace{2mm}
\begin{equation}\label{EstimationK1}
\hspace{-7cm}\leq\varepsilon C^{(1)} \max_{1\leq i\leq N} \Big\{c_{i,\infty}^{1-\delta}|u_i|\mathcal{R}_{\varepsilon, i}\Big\} \big\|\boldf\big\|_{\spaceLvM}^2.
\end{equation}

\vspace{2mm}
\noindent \textbf{Step 7 -- End of the proof.} Choosing $\bar{\delta}=\max\Big\{\frac{1}{2},\bar{\delta}_1,\bar{\delta}_2\Big\}$, estimates (\ref{EstimationK4})--(\ref{EstimationK2}) and (\ref{EstimationK1}) are simultaneously verified, and we can gather them together to obtain, for all $\boldf\in\spaceLMR$,

\[
\big\langle\mathbf{\epsL}(\boldf)-\boldL(\boldf),\boldf\big\rangle_{\spaceLM}\leq  C\varepsilon\max_{1\leq i\leq N} \Big\{c_{i,\infty}^{1-\delta}|u_i|\mathcal{R}_{\varepsilon, i}\Big\} \big\|\boldf\big\|_{\spaceLvM}^2,
\]
with $\displaystyle C=\max_{1\leq i\leq 4} C^{(i)}$.

\vspace{1mm}
 To conclude, we simply observe that we can write $\boldf =\boldf - \pi_{\boldL}(\boldf) + \pi_{\boldL}(\boldf)$ and then apply Young's inequality in this last estimate, to finally get

\[
\hspace{-1cm}\big\langle \mathbf{\epsL}(\boldf), \boldf\big\rangle_{\spaceLM}\leq -\left(\lambda_{\boldL}-C\varepsilon\max_{1\leq i\leq N} \Big\{c_{i,\infty}^{1-\delta}|u_i|\mathcal{R}_{\varepsilon, i}\Big\}\right)\|\boldf-\pi_{\boldL}(\boldf)\|_{\spaceLvM}^2
\] 

\[
\hspace{0.5cm}+C\varepsilon\max_{1\leq i\leq N} \Big\{c_{i,\infty}^{1-\delta}|u_i|\mathcal{R}_{\varepsilon, i}\Big\}\|\pi_{\boldL}(\boldf)\|_{\spaceLvM}^2,
\]
for all $\boldf\in\spaceLMR$.

\vspace{2mm}
 The estimate holds for any $\varepsilon >0$ and Theorem \ref{MainResult} is then proved.


\nocite{*}
\bibliographystyle{plain}
\bibliography{Bibliografia}

\noindent \scshape{A.B., M.B., B.G.: MAP5, CNRS UMR 8145, Sorbonne Paris Cit\ac{e}, Universit\ac{e} Paris Descartes, France}

\noindent \textit{E-mail address}: \textrm{\bf andrea.bondesan@parisdescartes.fr}

\noindent \textit{E-mail address}: \textrm{\bf marc.briant@parisdescartes.fr}

\noindent \textit{E-mail address}: \textrm{\bf berenice.grec@parisdescartes.fr}\\

\noindent \scshape{L.B.: Sorbonne Universit\ac{e}s, UPMC Univ Paris 06, CNRS, INRIA, Laboratoire Jacques-Louis Lions (LJLL CNRS UMR 7598), \ac{E}quipe-projet Reo, Paris, F-75005, France}

\noindent \textit{E-mail address}: \textrm{\bf laurent.boudin@upmc.fr}

\end{document}